\documentclass{article}

\usepackage{graphicx} 
\usepackage{orcidlink}
\usepackage[utf8]{inputenc}
\usepackage{mathtools,amsthm,amsmath,amssymb,comment}
\usepackage{breqn}
\usepackage{import}
\usepackage{comment}
\usepackage{bbm}
\usepackage{appendix}
\usepackage{thm-restate}
\usepackage{lipsum}
\usepackage{authblk}

\usepackage{hyperref}
\usepackage{tcolorbox}

\usepackage[normalem]{ulem}

\usepackage{tikz-cd}

\newtheorem{prop}{Proposition}[section]
\newtheorem{theorem}[prop]{Theorem}
\newtheorem{definition}[prop]{Definition}
\newtheorem{lemma}[prop]{Lemma}
\newtheorem{corollary}[prop]{Corollary}
\theoremstyle{remark}

\usepackage[isbn=false, doi=false, url=false, maxbibnames=99]{biblatex} 
\DeclareFieldFormat[article]{number}{\mkbibparens{#1}}

\renewbibmacro*{volume+number+eid}{%
  \printfield{volume}%
  \printfield{number}%
  \setunit{\addcomma\space}%
  \printfield{eid}}

\renewbibmacro*{in:}{%
  \ifentrytype{article}
    {}
    {\printtext{\bibstring{in}\intitlepunct}}}

\usepackage[a4paper, total={6in, 8in}]{geometry}

\newcommand{\WF}{\textup{WF}}
\DeclareMathOperator{\supp}{supp}
\newcommand{\sE}{\mathcal{E}}
\newcommand{\sD}{\mathcal{D}}
\newcommand{\sF}{\mathcal{F}}

\newcommand{\sX}{\mathcal{X}}

\newcommand{\dg}[1]{\mathcal{D}'_{\Gamma_{#1}}}
\newcommand{\el}[1]{\mathcal{E}'_{\Lambda_{#1}}}
\newcommand{\egc}[1]{\mathcal{E}'_{\Gamma^c_{#1}}}

\newcommand{\nbhd}{neighbourhood of zero }
\newcommand{\hatotimes}{\mathbin{\hat\otimes}}
\newcommand{\pc}{{\mathrm{pc}}}
\newcommand{\fc}{{\mathrm{fc}}}
\newcommand{\ec}{{\mathrm{ec}}}

\title{A novel class of functionals for perturbative algebraic quantum field theory}
\author{Eli Hawkins\thanks{\href{mailto:eli.hawkins@york.ac.uk}{eli.hawkins@york.ac.uk}\orcidlink{0000-0003-2054-3152}}}
\author{Kasia Rejzner\thanks{\href{mailto:kasia.rejzner@york.ac.uk}{kasia.rejzner@york.ac.uk}\orcidlink{0000-0001-7101-5806}}}
\author{Berend Visser\thanks{\href{mailto:berend.visser@york.ac.uk}{berend3.14@gmail.com}\orcidlink{0000-0002-5518-3807}}}
\affil{Department of Mathematics\\ The University of York\\ United Kingdom}
\date{}

\begin{document}
\maketitle
\begin{abstract}
Perturbative Algebraic Quantum Field Theory (pAQFT) is based upon formal power series valued in spaces of functionals. 
This is usually done with  \textit{microcausal} functionals, which are defined using microlocal analysis and motivated by propagation of singularities. In this paper, we prove that the class of microcausal functionals is not closed under the Peierls (Poisson) bracket by showing that a Peierls bracket of regular functionals can fail to be smooth. Consequently, microcausal functionals are not a suitable basis for pAQFT.
To remedy this issue, we introduce the class of \emph{equicausal functionals}. 
We show that this class contains the local functionals and that it closes under the $\star$-product and Peierls bracket. Furthermore, we prove the time-slice axiom for equicausal functionals, using a chain homotopy. 
 The class of microcausal functionals is not closed under this chain homotopy, which strongly suggests that the class of microcausal functionals does not fulfill the time slice axiom. 
\end{abstract}
\tableofcontents

\section{Introduction}
 The functional formalism for field theory  applies techniques from infinite-dimensional geometry to the rigorous description of field theories in a Lagrangian setting. It fully embraces the fact that the  configuration space of a field theory is infinite-dimensional. Typically, this is the space of sections of some vector bundle over spacetime, or the subspace of solutions to the Euler-Lagrange equations. 

 This formalism has been developed mainly in \cite{dutsch2001algebraic,dutsch2003master,dutsch2007action,duetsch2001perturbative,Brunetti2009} and some of the infinite dimensional differential geometry has been spelled out in \cite{Brunetti2019}, and, in the context of the quantization of gauge systems, in \cite{Fredenhagen2013,Fredenhagen2012}. It combines the idea of deformation quantization as the basis for the construction of free field theory with the Epstein-Glaser \cite{EG} renormalisation technique and microlocal analysis insights going back to Radzikowski \cite{Radzikowski1996}. For a review, see \cite{Rejzner2016} and \cite{dutsch2019classical}.
 
The central notion in this approach is that of a functional, i.e., a function on the configuration space. The derivatives of functionals are distributions, so the Poisson bracket and $\star$-product of arbitrary smooth functionals are ill-defined, as they require pairing distributions together. A \textit{microcausal} functional is one whose derivatives are required (in terms of microlocal analysis) to have a restricted singular structure, stemming from the causal structure of the spacetime (hence the name `microcausal'). This restriction is designed precisely to render the pairings in these algebraic operations well defined. This was made possible due to the seminal paper of Radzikowski, where the Hadamard condition on a quasifree state was phrased in terms of a microlocal condition on its two-point function \cite{Radzikowski1996}. It opened the way for applying microlocal analysis techniques in QFT and since then, the approach has been extremely successful both in Minkowski spacetime and on more general globally hyperbolic backgrounds in a number of important contributions \cite{Brunetti1996,BF0,BF97,HW,HW01,Baer2009,Brunetti2003}.

In this paper we show that the class of microcausal functionals, which is commonly used in the literature, suffers from an unexpected and surprising defect: Although the Peierls bracket of two arbitrary microcausal functionals is a well defined functional, it is not always continuous (and hence not microcausal).
We prove this with an explicit example of two regular (hence microcausal) functionals whose Peierls bracket is discontinuous.
The fact that this problem has gone unnoticed for the past decade is presumably due to the fact that most authors have restricted their attention to \emph{polynomial functionals}. Physically, this restriction rules out many interesting theories, e.g.\ the sine-Gordon model (where the interaction is a cosine) and many relevant observables, e.g.\ gauge-invariant observables in Yang-Mills theories that model Wilson loops. 

We have found a related issue in the connection between on-shell and off-shell algebras in the functional formalism. In the on-shell approach, observables are  functionals on the space of solutions to the Euler-Lagrange equations, i.e.\ `on the mass-shell'. 
The modern approach,  using homological algebra (see e.g.\ \cite{Fredenhagen2012}),  is to work off shell and  study a \textit{(co-)chain complex} of multivector fields on the space of \textit{all} configurations (solutions or not). The space of on-shell observables is then defined as the (degree $0$) homology of the complex. 

It has been a `folklore theorem' for some time that, for free field theories, these approaches are equivalent. 
We prove this for the complex of arbitrary smooth multivector fields using a natural chain homotopy. This also implies the time-slice axiom. The microcausal complex is \textit{not} closed under this homotopy, so this proof fails there and strongly suggests that the homology of the microcausal complex is not equivalent to a space of on-shell functionals. This also indicates that the time-slice axiom fails for the microcausal complex.

The reason for all of these problems is the pointwise nature of the  definition of  microcausal functionals. To remedy this, we define \textit{equicausal} functionals in Definition \ref{equicausaldefinition}. These are microcausal functionals with a further boundedness condition on the way that the singular structure varies over the configuration space. This uses a version of Hörmander's topology on spaces of distributions with a prescribed wavefront set and the notion of \textit{equicontinuous} sets.

We show that fundamental operations like differentiation and integration preserve the class of equicausal functionals. From this, we show that equicausal functionals are closed under  the $\star$-product, and hence the Peierls bracket. We also show that the complex of equicausal multivector fields is closed under the homotopy mentioned above, so its homology is equivalent to a space of on-shell functionals and the time-slice axiom is satisfied.

This paper is structured as follows: In Section \ref{Background} we recall some background on functional analysis, i.e.\ calculus on infinite-dimensional spaces, and its applications in the functional formalism. We also give a (reasonably) self-contained summary of the microlocal techniques we use in this paper; in particular we summarize the topology and properties of spaces with prescribed wavefront sets, as we feel these topics are non-standard. In Section \ref{KoszulComplex}, we perform a detailed study of the Koszul complex of smooth multivector fields (without microlocal constraints), and show that it is a resolution of the algebra of on-shell functionals. As a consequence, we prove the time-slice axiom for multivector fields. In Section \ref{pathologicalfeatures}, we show that a Poisson bracket of microcausal functionals is in general not a smooth functional by giving an explicit counterexample. In Section \ref{equicausal} we give the definition of equicausal multivector fields, and prove some basic properties. We also give some examples by showing that both local functionals and Wick polynomials are equicausal. In Section \ref{timeslice}, we prove the time-slice axiom for equicausal multivector fields. Finally, in Section \ref{algebraicstructsection} we show that the $\star$-product closes on equicausal functionals, implying in particular that the Poisson bracket also closes, as it is the first order part of the $\star$-commutator. We close in Section \ref{outlook} by considering  alternatives to our definition, and indicate some future research directions. The appendices contain several technical proofs that we felt overburdened the main text.

\section{Background}\label{Background}
In this section we gather some background on smooth functionals, and on microlocal analysis. Experts in the field can likely skip most or all or these subsections. We do not provide proofs for most facts, but rather give a citation at the start of each section for the interested reader to follow. 

\subsection{Equicontinuous sets}
As the concept of equicontinuity is central in our discussion, we take some time here to remind the reader of the definition and some basic consequences. Throughout this subsection, $A,B$ and $C$ are fixed locally convex topological vector spaces. A good source for this material is chapter 32 of \cite{Treves1967}.

\begin{definition}\label{equicontinuous}
A set $H$ of linear maps $A\to B$ is \textbf{equicontinuous} if, for any \nbhd $V\subset B$, there is a \nbhd $U \subset A$ such that
    \begin{equation*}
        f(U) \subset V \, \forall \, f \in H.
    \end{equation*}
\end{definition}
We will also somewhat abusively say that `$H$ is equicontinuous from $A$ to $B$'. One should note that, in particular, any individual $f \in H$ will be continuous. An equicontinuous set of mappings is thus a collection of mappings that are `continuous at the same rate'. 

In the case of maps into $\mathbb{C}$, which will be most relevant to us, there is a particularly simple characterisation of equicontinuous sets: A set of continuous linear functionals $H \subset A'$ is equicontinuous, iff there is a continuous seminorm $p$ of $A$, such that
\begin{equation*}
    |f(x)| < p(x) \, \forall \, x \in A,\, f \in H.
\end{equation*}
As an illustration of the central importance of equicontinuous sets in functional analysis, we recall the following fact. 
\begin{prop}
    The topology of a topological vector space $A$ is equivalent to the topology of uniform convergence on equicontinuous subsets of $A'$, the continuous linear dual of $A$. 
\end{prop}
This is often a useful viewpoint to take, as it allows one to move questions about convergence or continuity in $A$ to a similar question in $\mathbb{C}$. Unless explicitly mentioned otherwise, we take $A'$ to denote the \textit{strong} dual of $A$, i.e.\ the set of continuous linear maps $A \to \mathbb{C}$, equipped with the topology of uniform convergence on bounded sets of $A$.

We mention  some stability properties of equicontinuous sets. First off, if $H_1$ and $H_2$ are equicontinuous sets of linear maps $A\to B$, then their sum
\begin{equation*}
    H_1+H_2 = \{ f_1+f_2 \: |\: f_i\in H_i \} 
\end{equation*}
is again an equicontinuous set. Any subset of an equicontinuous set is also equicontinuous, as is any finite union of equicontinuous sets.\footnote{These facts can be summarised by saying that the equicontinuous subsets of $L(A,B)$ form a \textit{vector bornology}. This is in general a weaker bornology than the von Neumann bornology induced by the strong topology. It is the statement of the Banach-Steinhaus theorem that these bornologies match when $A$ is barreled and $B$ is locally convex. We refer the interested reader to \cite{HogbeNlend1977} for details on bornology.}
 
Furthermore, if $H$ is  equicontinuous from $A$ to $B$ and $L$ is equicontinuous from $B$ to $C$, then their composition is defined as
\begin{equation*}
    L \circ H = \{g\circ f \: | \: g \in L \: , \: f \in H \},
\end{equation*}
and this is equicontinuous from $A$ to $C$. In particular, by taking $L$ or $H$ to be a set containing a single continuous map, this implies that equicontinuous sets are stable under composition with continuous maps.
\subsection{Smooth functionals}
Another fundamental notion in the following is that of smooth functionals on locally convex topological vector spaces. As this notion is central to the theory we develop, we will gather some definitions and results in this section. Our main sources for these sections are \cite{Hamilton1982,Kriegl1997}, but we also recommend appendix A of \cite{Brunetti2019} for a summary that is tailored to our situation.
 
Let $A$ and $B$ denote locally convex topological vector spaces, and $F:A \to B$ a functional. We define the directional derivative of $F$ at $\varphi\in A$ in the direction $\psi \in A$ to be
\begin{equation*}
   \frac\delta{\delta\varphi}F(\varphi)\{\psi\} = 
    F^{(1)}(\varphi)\{\psi\} = \lim_{t\to 0}\frac{1}{t}\left(F(\varphi +t \psi)-F(\varphi)\right),
\end{equation*}
whenever this limit exists. A functional is continuously differentiable, or $C^1$, if its derivatives exist at all points and in all directions and define a \textbf{jointly} continuous map
\begin{align*}
    A \times A &\to B \\
    (\varphi,\psi) &\mapsto F^{(1)}(\varphi)\{\psi\}
\end{align*}
Equivalently, it can be shown that $F$ is $C^1$ if and only if there exists a continuous map $L: A\times A \times \mathbb{R} \to B$ such that
\begin{equation*}
    L(\varphi,\psi,t) = \begin{cases}
      \frac{1}{t}\left(F(\varphi +t \psi)-F(\varphi)\right) &t\neq 0, \\
      F^{(1)}(\varphi)\{\psi\} &t=0.
    \end{cases}
\end{equation*}
A similar statement holds for the partial derivative.

Similarly, the $n$-fold directional derivative is defined by 
\begin{equation*}
    F^{(n)}(\varphi)\{\psi_1,\ldots,\psi_n\} = \lim_{t_i\to 0}\frac{1}{t_1\ldots t_n}\left(F\left(\varphi +\sum_i t_i \psi_i\right)-F(\varphi)\right),
\end{equation*}
whenever the limit exists, and call a functional $C^n$ if all these quantities exist and form a continuous map $A\times A^n \to B$. A functional is \textbf{Bastiani} smooth if it is $C^n$ to all orders \cite{Bas64}. The following facts hold for such functionals:
\begin{itemize}
    \item $F$ is continuous.
    \item $F^{(n)}: A\times A^n \to B$ is multi-linear in its last $n$ arguments, which we typographically indicate by putting all linear arguments in braces rather than parentheses:
    \begin{equation*}
        F^{(n)}: (\varphi,\psi_1,\ldots,\psi_n) \mapsto F^{(n)}(\varphi)\{\psi_1,\ldots,\psi_n\}.
    \end{equation*}
    \item $F^{(n)}$ is the partial derivative of $F^{(n-1)}$ with respect to its nonlinear argument:
    \begin{equation*}
        F^{(n)}(\varphi)(\psi_1,\ldots,\psi_n) = \frac{\delta}{\delta\varphi}F^{(n-1)}(\varphi)\{\psi_1,\ldots,\psi_{n-1}\}\{\psi_n\}.
    \end{equation*}
    \item The fundamental theorem of calculus holds:
    \begin{equation*}\label{ftcalc}
        F(\varphi_0+\varphi_1)-F(\varphi_0) = \int_0^1  F^{(1)}(\varphi_0 +\lambda \varphi_1)\{\varphi_1\}d\lambda.
    \end{equation*}
    \item The chain-rule holds: If $F$ and $G$ are two composable smooth maps, then $G\circ F$ is also smooth and
    \begin{equation*}
        (G \circ F)^{(1)}(\varphi)\{\psi\} = G^{(1)}(F(\varphi))\left\{F^{(1)}(\varphi)\{\psi\}\right\}.
    \end{equation*}
    \item If $A=A_1 \times A_2$ is a product of vector spaces, then $F$ is smooth iff all of its partial derivatives exist and are continuous.
\end{itemize}

The notion of Bastiani smoothness is not the only one available in the literature. A functional $F:A\to B$ is \textbf{conveniently} smooth if, whenever $\gamma:\mathbb{R} \to A$ is a smooth curve, $F\circ \gamma: \mathbb{R} \to B$ is also smooth. This notion of smoothness is in general weaker than Bastiani smoothness, but is of course easier to check as curves are easier to handle than general nonlinear functionals. It was shown in Theorem 1 of \cite{Froelicher1982} that this notion of smoothness coincides with Bastiani smoothness if $A$ is a metric space and $B$ is complete. In this text, \textbf{smooth} will always mean Bastiani smooth.

Through the fundamental theorem of calculus, integration in topological vector spaces plays an important role in our constructions. Recall that, if $a:\mathbb{R} \to A$ is a continuous curve in a topological vector space, and $c\leq d$, then the integral 
\begin{equation*}
    \int_c^d a_t dt 
\end{equation*}
exists as an element of the completion $\hat{A}$. Furthermore, if $l:A \to B$ is a continuous linear map, then the equality
\begin{equation*}
   l \left(\int_c^d a_t \, dt\right) =  \int_c^d l (a_t) \, dt
\end{equation*}
holds in $\hat{B}$, where $l$ also denotes the map on the completions $\hat{A} \to \hat{B}$. We refer the reader to the literature for details, see e.g.\ Chapter 2 in \cite{Kriegl1997}. Throughout this text, we will need to integrate functionals depending on a real parameter. To show that this yields another smooth functional, we will use the following version of the Leibniz integral rule.
\begin{restatable}[Leibniz integral rule]{prop}{LeibnizIntegralRule}
\label{Leibnizintegralrule}
    Let $A,B$ be locally convex topological vector spaces, of which $B$ is complete. If $F:\mathbb{R}\times A\to B$ is a smooth functional, then the functional $G$ defined by
    \begin{equation*}
        G(\varphi) = \int_0^1 F(s,\varphi) ds
    \end{equation*}
    is smooth. Furthermore, its derivatives are given by
    \begin{equation*}
        G^{(n)}(\varphi)\{h_1,\ldots, h_n\} = \int_0^1 \frac{\delta^n}{\delta \varphi^n}F(s,\varphi) \{h_1,\ldots,h_n\} ds
    \end{equation*}
\end{restatable}
As we are not aware of a proof in the literature in the generality that we require, we provide one in Appendix \ref{integralruleappendix}.

Finally, we discuss an alternative way of viewing the derivatives of a functional $F:A\to B$. \textit{A priori} $F^{(n)}$ is a jointly continuous map $A\times A^{\times n} \to B$, linear in the last $n$ arguments. Fixing $\varphi\in A$, it defines a map
\begin{equation}\label{partialevaluation}
    F^{(n)}(\varphi): (a_1,\ldots,a_n) \mapsto F^{(n)}(\varphi)\{a_1,\ldots,a_n\}
\end{equation}
that is $n$-linear and continuous. By definition, this is the same thing as a continuous linear map from the completed projective tensor product $A^{\hatotimes_\pi n}$ to $B$. 

Let $L_c(A^{\hatotimes_\pi n},B)$ denote the space of all such maps, equipped with the topology of uniform convergence on compact sets. We recall a general fact: Let $C$ and $D$ be topological vector spaces, and $X$ a topological space. If $f:X\times C \to D$ is a jointly continuous map that is linear in $C$, then $f$ defines a continuous map
    \begin{equation*}
        \bar{f}: X \to L_c(C,D),
    \end{equation*}
by
\begin{equation*}
    x\mapsto f(x,\_)=\left(c\mapsto f(x,c)\right).
\end{equation*}
This is nothing more than the observation that the topology on $L_c(C,D)$ is the restriction of the compact-open topology on $C^0(C,D)$, which has this property.
\begin{prop}\label{derissmoothfunctional}
    If $F:A\to B$ is a smooth functional and $n\in\mathbb{N}$, then the map
\begin{align*}
    \bar{F}^{(n)}:A &\to L_c(A^{\hatotimes_\pi n},B) \\
    \varphi &\mapsto F^{(n)}(\varphi)
\end{align*}
defined by equation \eqref{partialevaluation} is smooth. 
\end{prop}
\begin{proof}
The previous lemma implies that $\bar{F}^{(n)}$ is continuous for all $n$. The partial derivative of $F^{(n)}$ with respect to $\varphi$ equals $F^{(n+1)}$. Hence there exists a map $L$, linear in its last argument
    \begin{equation*}
        L:A\times A \times \mathbb{R} \times A^{\hatotimes_\pi{n}} \to B
    \end{equation*}
    such that for $t \neq 0$ 
    \begin{equation*}
       L(\varphi,\psi,t,g) = \begin{cases}
           \frac{1}{t}\left(F^{(n)}(\varphi + t \psi)\{g\} - F^{(n)}(\varphi)\{g\}\right) &t\neq0, \\
           F^{(n+1)}(\varphi)\{\psi;g\} &t=0.
        \end{cases}
    \end{equation*}  
This means that the map
\begin{align*}
    \overline{L}:\:&
        A\times A \times \mathbb{R} \to L_c(A^{\hatotimes_\pi n},B), \\
    &(\varphi,\psi,t) \mapsto L(\varphi,\psi,t,\_)
\end{align*}
is continuous. Hence it follows that
\begin{dmath*}
    \frac{\delta}{\delta \varphi} \bar{F}^{(n)}(\varphi)\{\psi\} =  \lim_{t\to 0} \frac{1}{t}\left(F^{(n)}(\varphi+t\psi)\{\_\}-F^{(n)}(\varphi)\{\_\}\right) = \lim_{t\to 0} \overline{L}(\varphi,\psi,t) = \overline{L}(\varphi,\psi,0) = F^{(n+1)}(\varphi)\{\psi,\_\},
\end{dmath*}
which is jointly continuous in $\varphi$ and $\psi$, again using the previous lemma. Similarly, we obtain
\begin{equation*}
    \frac{\delta^m}{\delta \varphi^m} \bar{F}^{(n)}(\varphi)\{\psi_1,\ldots, \psi_m\} = {F}^{(n+m)}(\varphi)\{\psi_1,\ldots,\psi_n,\_\}
\end{equation*}
so that $\bar{F}^{(n)}$ is continuously differentiable to all orders, again using the observation preceding this proposition.
\end{proof}
In what follows, we will drop the bar from the notation, and use $F^{(n)}$ for both maps interchangeably.

\subsection{Some Microlocal Analysis}
An important tool for the study of perturbative AQFT is microlocal analysis. This studies the singular structure of distributions, and describes pairings of distributions that would naively be ill-defined. An extensive discussion of this subject would take us too far afield, so we recall some facts that might not be common knowledge, and refer the reader to the literature for more details. For introductory treatments, we recommend \cite{Brouder2014}, as well as chapter 4 in \cite{Baer2009}. The standard reference for this subject is \cite{hormander2015analysis}.

Let $M$ be a manifold, we use the notation $\sD(M)$ for the smooth compactly supported functions on $M$,  $\sE(M)$ for the smooth functions on $M$, $\sD'(M)$ for the distributions on $M$, and $\sE'(M)$ for the compactly supported distributions on $M$. The space $\sE(M)$ has a natural topology, which can be described as the topology of uniform convergence of all derivatives on compact subsets of $M$. This turns it into a Fréchet space. The space $\sD(M)$ can then be exhibited as the inductive limit of $\sE(K)=\{\varphi \in \sE(M) \, | \, \supp(\varphi) \subset K\}$ for $K\subset M$ compact, and is given the inductive limit topology with respect to this diagram, turning it into an LF-space. Finally, the spaces $\sD'(M)$ and $\sE'(M)$ are endowed with their strong dual topologies. We refer the reader to \cite{Treves1967} for details on these topologies.

Recall the definition of distributional sections:
\begin{definition}
 Let $E\to M$ be a vector bundle. We denote by $\sE(M;E)$ the space of smooth sections of $E$, and by $\sD(M;E)$ the space of compactly supported sections. Similar to the scalar case, $\sE(M;E)$ has a natural Fréchet topology, and $\sD(M;E)$ is endowed with the inductive topology stemming from the inclusions of $\sE(E|_K;K)$ for compact $K\subset M$, turning it into an LF-space.
The \textbf{functional dual} of $E$ is the bundle
\begin{equation*}
    E^! = E^* \otimes D_M
\end{equation*}
where $D_M$ is the density bundle of $M$.

The distributional sections of $E$ are defined by
    \begin{equation*}
     \sD'(M;E)  \equiv \sD(M;E^!)'.
    \end{equation*}
Similarly, the compactly supported distributional sections are defined by
    \begin{equation*}
   \sE'(M;E)  \equiv \sE(M;E^!)'.
    \end{equation*}
We will always equip these spaces with their strong dual topology. 
\end{definition}

The central notion in microlocal analysis is that of the wavefront set of a distribution. This is a sharpening of the notion of singular support, in that it not only describes \textit{where} a distribution is singular, but also in what direction in Fourier space. Careful treatment of wavefront sets of distributions allows us to multiply distributions, pull them back, push them forward and evaluate them against each other, all as extensions of those concepts on smooth functions. In what follows, a \textbf{cone} in $\mathbb{R}^n\setminus\{0\}$ is a subset that is stable under multiplication by positive constants.
\begin{definition}
     Let $\Omega \subset \mathbb{R}^n$ be non-empty and open, and let $u\in \sD'(\Omega )$. A pair $(x,\xi) \in \Omega \times \mathbb{R}^n\setminus \{0\}$ is a \textbf{regular directed point} of $u$, if there exist 
     \begin{itemize}
         \item A conic neighbourhood $V \subset \mathbb{R}^n\setminus \{0\}$ of $\xi$,
         \item A test function $f\in \sD(\mathbb{R}^n)$ such that $f(x)\neq 0$,
     \end{itemize}
such that, for any $N\in \mathbb{N}$,
\begin{equation*}
    \sup_{\xi\in V}(1+|\xi|)^N |\widehat{fu}(\xi)| < \infty,
\end{equation*}
where a hat denotes the Fourier transform. The \textbf{wavefront set} $\WF(u)$ is the complement in $\mathbb{R}^n \times \mathbb{R}^n\setminus \{0\}$  of the set of regular directed points of $u$.
\end{definition}

We note immediately that, because the condition to be in the complement of the wavefront set is phrased in terms of conic \textit{neighbourhoods}, the wavefront set of a distribution is a closed set. Furthermore, it can be shown that it behaves like a subset of the cotangent bundle under diffeomorphisms. If $\psi: \Omega_1 \to \Omega_2$ is a diffeomorphism and $u\in\sD'(\Omega_2)$, then the pullback $\psi^* u \in \sD'(\Omega_1)$ is well-defined, as $\psi$ is in particular a submersion. If $d\psi$ denotes the differential of $\psi$, then
\begin{equation*}
    \WF(\psi^*u) = d\psi^* \WF(u)
\end{equation*}
see e.g.\ Theorem 8.2.4 and following discussion in \cite{hormander2015analysis}. This means that  the definition of wavefront set naturally extends to distributions on manifolds: If $v\in \sD'(M)$, then 
\begin{equation*}
  (x,\xi)\in \dot{T}^*M = \{ (y,\eta) \in T^*M \: | \: \eta\neq 0 \}  
\end{equation*}
is an element of $\WF(v)$ if and only if $(x,\xi) \in WF(v|_U)$ for some (and hence any) coordinate chart $U$ around $x$, where we have suppressed the coordinate function in the notation.

One can push the definition further still to include distributional sections of vector bundles. If $U\subset \mathbb{R}^n$ is an open region, and $\Vec{u}$ a distributional section of $\mathbb{R}^k \times U$, i.e.\ a $k$ component vector $(u_1, \ldots, u_k)$ of distributions on $U$, then we define
\begin{equation*}
    \WF(\Vec{u}) \equiv \bigcup_i \WF(u_i).
\end{equation*}
The transition functions of a vector bundle locally look like matrices of smooth functions. As multiplication by a smooth function leaves the wavefront set of a distribution invariant, we see that the definition lifts from this local model to arbitrary vector bundles.

We will be interested in spaces of distributions with prescribed wavefront sets. A \textbf{cone} $\Gamma\subset\dot{T}^*M$ is a set such that $(x,\xi)\in \Gamma \implies (x,\lambda\xi) \: \forall \: \lambda>0$. In particular, wavefront sets are cones. If $\Gamma \subset \dot{T}^*M$ is a closed cone and $\Lambda \subset \dot{T}^*M$  is an open cone, we introduce the following two classes of distributional sections:
    \begin{align*}
        \dg{}(M;E) &= \{u\in\sD'(M;E) \: | \: \WF(u) \subset \Gamma \}, \\
        \el{}(M;E) &= \{v\in\sE'(M;E) \: | \: \WF(v) \subset \Lambda \}.
    \end{align*}
When considering the trivial bundle $M\times \mathbb{R}\to M$, we will instead denote them by $\dg{}(M)$ and $\el{}(M)$. 

These spaces are topologised as follows, where we treat the scalar case first. The topology on $\sD'_\Gamma(M)$ is called the \textit{normal} topology, due to the fact that it turns $\dg{}(M)$ into a normal space of distributions. It is defined as follows:
\begin{definition}
    Let $M$ be a manifold and $\Gamma \subset \dot{T}^*M$ a closed cone. Let $\chi$ be a test function in $\sD(M)$, supported inside a coordinate patch so that Fourier transforms are well-defined and the cotangent bundle is trivial,\footnote{Again, we suppress the implicit choice of coordinate functions in the notation.} $N\in\mathbb{N}$ and $V$ a closed cone with $\textup{supp}(\chi)\times V\cap\Gamma = \emptyset$ and $B\subset \sD(M)$ is a bounded set. Then we define the following seminorms on $\dg{}(M)$:
    \begin{align*}
    P_{\chi,N,V}(u) &= \sup_{\xi \in V} (1+|\xi|)^N |\widehat{\chi u}(\xi)|,\\
    P_B(u) &= \sup_{f \in B} |u(f)|,
\end{align*}
 The \textbf{normal topology} on $\dg{}(M)$ is defined by these seminorms for all admissible choices of $\chi,N,V$ and $B$.
\end{definition}
For those familiar with it, this is nothing but Hörmander's topology, but using the seminorms of the strong topology on $\sD'(M)$ rather than those of the weak topology. We will \textbf{always} equip spaces of the form $\dg{}(M)$ with this topology. 

The topology of $\el{}(M)$ is a bit more involved, and rests on the following observation.  Due to Hörmander's criterion for multiplication of distributions, we can multiply $u\in \sD'_{-\Lambda^c}(M)$ and $v\in \el{}(M)$ to obtain another distribution $u v$. As this is a compactly supported distribution, we can evaluate it on the function that is $1$ everywhere:
\begin{equation*}
    \langle u, v\rangle = u v(1).
\end{equation*}
In \cite{Dabrowski2014}, Brouder and Dabrowski showed that, when $\sD'_{-\Lambda^c}(M)$ is endowed with the normal topology described above (note that $-\Lambda^c$ is closed as $\Lambda$ is open), this pairing turns $\el{}(M)$ into its dual. As such, we endow $\el{}(M)$ with the strong dual topology stemming from this duality. We mention here that is somewhat poorly behaved from the viewpoint of functional analysis. For example, as $\sD_{-\Lambda^c}$ is not a barrelled space, the Banach-Steinhaus theorem does \textbf{not} apply, meaning the equicontinuous subsets of $\el{}$ are a strict subset of the bounded sets. Furthermore, $\el{}$ is \textbf{not} complete.

As $\el{}(M)$ is the dual of $\sD'_{-\Lambda^c}(M)$, we can consider equicontinuous subsets $H\subset\el{}(M)$. The following characterisation of such sets, which is Lemma 5.3 of \cite{Brouder2016}, is of a fundamental importance to our investigation:
\begin{prop}\label{characterizationequicontinuous}
    Let $\Lambda$ be an open cone inside $\dot{T}^*M$. If $H\subset\sE'_\Lambda(M)$ is equicontinuous, then there exist a compact set $K \subset M$ and a \textbf{closed} cone $\Xi \subset \Lambda$ such that
\begin{itemize}
    \item $H$ is a subset of $\mathcal{D}'_\Xi(K)=\{u\in \mathcal{D}'_\Xi(M)\: | \: \supp(u)\in K\} $,
    \item $H$ is bounded in the topology that $\mathcal{D}'_\Xi(K)$ inherits as a subset of $\mathcal{D}'_\Xi(M)$.
\end{itemize}
\end{prop}
The tensor product of distributions will become important once we consider algebraic operations using distributions. Suppose that $N$ and $P$ are manifolds. Throughout this text, we (abusively) denote by $\underline{0}$ the zero section of different cotangent bundles, trusting that it is clear which one we mean from the context. If $u\in \sD'(N)$ and $v\in \sD'(P)$ then
\begin{equation*}
    \WF(u\otimes v) \subset \left[\WF(u)  \times \WF(v)\right] \cup \left[ \underline{0}  \times \WF(v)\right] \cup \left[\WF(u) \times \underline{0}\right].
\end{equation*}
Hence we introduce the following product of cones: 
\begin{definition}
    Let $\Gamma_a \subset \dot{T}^*N $ and $\Gamma_b\subset \dot{T}^*P$ be  cones. Their \textbf{dotted product} is defined by
\begin{equation}\label{dotproduct}
    \Gamma_a \dot{\times} \Gamma_b \equiv\left[\Gamma_a \times \Gamma_b\right] \cup \left[\underline{0} \times \Gamma_b\right] \cup \left[\Gamma_a \times \underline{0}\right] \subset \dot T^*(N\times P). 
\end{equation}
\end{definition}
We note that alternatively, we could have written
\begin{equation*}
    \Gamma_a \dot{\times} \Gamma_b = (\Gamma_a\times \underline{0}) \times (\underline{0} \times \Gamma_b) \setminus \underline{0}\times \underline{0}
\end{equation*}
Furthermore, as the zero section of a vector bundle is closed, the product of two closed cones is again closed. The product of open cones is in general not open. Recall the following general definition:
\begin{definition}
   Let $A,B$ and $C$ be topological vector spaces, a bilinear map $\Phi:A\times B \to C$ is \textbf{hypocontinuous} if, whenever $X\subset A$ and $Y\subset B$ are bounded sets, the sets 
    \begin{align*}
        \{\Phi(x,\_) \mid x \in X \} \subset L(B,C), \\
        \{\Phi(\_,y) \mid y \in Y\} \subset L(A,C),
    \end{align*}
are equicontinuous.
\end{definition}
Hypocontinuity is a weaker condition than joint continuity, but stronger than joint sequential continuity. Many bilinear operations on topological vector spaces fail to be continuous, but are hypocontinuous. The following fact is Theorem 4.6 in \cite{Brouder2016}:
\begin{prop}\label{hypocontinuous}
    Let $N$ and $P$ be manifolds and let $\Gamma_a \subset \dot{T}^*N $ and $\Gamma_b\subset \dot{T}^*P$ be closed cones. The tensor product of distributions restricts to a hypocontinuous map
\begin{align}\label{tensorprodhypo}
    \dg{a}(N) \times \dg{b}(P) \to \sD_{\Gamma_a \dot\times \Gamma_b}(N \times P ).
\end{align}
\end{prop}
It is a consequence of the preceding fact that the tensor product is continuous on spaces defined with respect to open cones. Note that, in contrast to the closed cone scenario, there is no `smallest' space to map into, as $\Lambda_a\dot\times \Lambda_b$ fails to be open in general.  
\begin{restatable}[]{prop}{tensorproductehypo}\label{tensorproductehypo}
Let $N$ and $P$ be manifolds, $\Lambda_a\subset\dot{T}^*N$ and $\Lambda_b \subset \dot{T}^*P$  open cones and $\Lambda\subset \dot T^* (N\times P)$ an open cone containing $\Lambda_a\dot\times\Lambda_b$. The tensor product of distributions restricts to a hypocontinuous map
\begin{equation}
     \sE'_{\Lambda_a}(N) \times\sE'_{\Lambda_b}(P) \to \sE'_{\Lambda}(N \times P).
\end{equation}
\end{restatable}
As we are not aware of a proof of this fact in the literature, we give one in Appendix \ref{technicallemmata}. We combine the previous facts into a technical lemma that will be instrumental in showing closure of the $\star$-product for equicausal functionals in Section \ref{algebraicstructsection}. 
\begin{prop}\label{extensiontensorprod} In the same situation as in the previous proposition, if $H\subset \sE'_{\Lambda_a}(N)$ and $L\subset \sE'_{\Lambda_b}(P)$ are equicontinuous subsets, then $H\otimes L$ is an equicontinuous subset of $\sE'_{\Lambda}(N\times P)$. Alternatively phrased, $H\otimes L$ defines an equicontinuous set of linear mappings
\begin{equation*}
    \sD'_{\Gamma}(N\times P) \to \mathbb{C}
\end{equation*}
for any closed cone $\Gamma$ satisfying
\begin{equation*}
    (\Lambda_a \dot\times \Lambda_b)\cap -\Gamma = \emptyset. 
\end{equation*}
\end{prop}
\begin{proof}
  Write $H \subset \sD'_{\Xi}(K) $ and $L\subset \sD'_{\Delta}(Q)$ for $\Xi,\Delta$ closed cones and $K,Q$ compact. The tensor product of distributions is hypocontinuous as a map
\begin{equation*}
  \sD'_\Xi(K) \times \sD'_\Delta(Q) \to  \sD'_{\Xi\dot{\times}\Delta}(K\times Q).  
\end{equation*}
As the image of a product of bounded sets by a hypocontinuous map is bounded, see e.g.\ \cite{Horvath1966}, we find that $H\otimes L$ is bounded in $\sD'_{\Xi\dot{\times}\Delta}(K\times Q)$, and hence is equicontinuous as a subset of $\sE'_{\Lambda}(N\times P)$. 

For the final statement, we note that $\left(\Lambda_a \dot\times\Lambda_b\right)\cap -\Gamma = \emptyset$ implies that $H\otimes L$ is an equicontinuous subset of $\sE'_{-\Gamma^c}$, so that it is an equicontinuous set of linear functionals on $\sD'_{-(-\Gamma^c)^c}=\dg{}$.
\end{proof}

We close this section by briefly discussing the bundle-valued case. Let $U_i$ be a locally finite open cover of $M$, trivializing $E$. Then the trivialisations gives isomorphisms
\begin{equation*}
    \dg{}(U_i;E|_{U_i}) \cong \mathcal{D}'_{\Gamma_{i}}(U_i)^k,
\end{equation*}
where $\Gamma_i = \Gamma|_{U_i}$ and $k$ is the rank of $E$. A partition of unity argument shows that the diagram 
\begin{equation*}\label{colimitdg}
    \begin{tikzcd}
        &\mathcal{D}'_{\Gamma_{i}}(U_i)^k \ar[r] & \mathcal{D}'_{\Gamma_{ij}}(U_i \cap U_j)^k \ar[<->,dd] \\
\dg{}(M \text{;}\, E) \ar[ur] \ar[dr] &&\\
&\mathcal{D}'_{\Gamma_{j}}(U_j)^k \ar[r] & \mathcal{D}'_{\Gamma_{ji}}(U_j \cap U_i)^k 
    \end{tikzcd}
\end{equation*}
is a colimit diagram of vector spaces, where the vertical arrow is a transition function of the bundle $E$. We endow $\dg{}(M;E)$ with the initial locally convex vector topology with respect to this diagram, making it a colimit diagram in LCTVS. It follows from a standard argument that this topology does not depend on the choice of cover.

The results quoted in this section have straightforward generalisations to the bundle valued case. For example, one can show that
\begin{equation*}
    (\dg{}(E))' \cong \sE'_{-\Gamma^c}(E^!),
\end{equation*}
and that the tensor product defines a hypocontinuous map 
\begin{align*}
    \dg{a}(E_a) \times \dg{b}(E_b) \xrightarrow{\otimes} \sD_{\Gamma_a \dot\times \Gamma_b}(E_a\boxtimes E_b).
\end{align*}
and similar for open cones. Here $E_a\boxtimes E_b\to N\times P$ denotes the exterior tensor product, whose fibre over $(x,y)\in N \times P$ is given by $E_{a,x}\times E_{b,y}$. All of these facts follow from combining the relevant fact for the scalar case with the diagram \eqref{colimitdg}.  
\subsection{Functionals on spacetime}
We now specialise to the case of interest for perturbative algebraic quantum field theory. We consider a spacetime $M$, equipped with a vector bundle $E$ that describes the field content of the theory. We take this to be fixed for the remainder of this section. The \textit{configuration space} of the theory is the space of smooth sections of this bundle. To abbreviate notation, we write $\sE\equiv\sE(M,E)$, and $\sE^!=\sE(M,E^!)$.

We will use functionals on $\sE$ to model observables of a field theory. These functionals have a notion of support on $M$ (in addition to the support as a map on $\mathcal{E}$, which is  less useful in this context). 
\begin{definition}
   The \textbf{spacetime support} of a functional $F$ is defined through its complement as 
\begin{align*}
    \textup{supp}(F)^c = \{ x \in M \: | \: \exists U \textup{ open neighbourhood of $x$ such that } \varphi|_{U^c}= \tilde{\varphi}|_{U^c} \implies F(\varphi) = F(\tilde{\varphi}) \}.
\end{align*}
\end{definition}
Roughly speaking, the support of a functional is that region of $M$ where $F$ is sensitive to perturbation. We mention the following characterisation of support from Lemma III.3 in \cite{Brouder2018}:
\begin{lemma}\label{derivativessupport}
    The support of a functional $F$ can be characterised as
    \begin{equation*}
        \supp(F) = \overline{\bigcup_{\varphi\in\sE}\supp(F^{(1)}(\varphi))}
    \end{equation*}
\end{lemma}
Note that, even though the $\supp(F^{(1)}(\varphi))$ is compact for all individual $\varphi\in\sE$, the union over all $\varphi$ is in general non-compact. We denote the space of all compactly supported smooth functionals on $\sE$ by $\mathcal{F}(E)$ or by $\sF$ when no confusion can arise. 

We introduce also the notion of multivector field. 
\begin{definition} 
A \textbf{$k$-vector field} is a smooth functional 
\begin{equation*}
    X: \sE \to  \sE'_a(M^k,E^{\boxtimes k}),
\end{equation*}
where the subscript $a$ indicates the subspace of completely antisymmetrised distributional sections of $\sE'(M^k,E^{\boxtimes k})$. 
\end{definition}

For each $\varphi$, the distributional section $X(\varphi)$ has an intrinsic notion of support on spacetime, in addition to the notion of support it has as a functional on $\sE$. The following definition takes both of these notions of support into account.
\begin{definition}
The \textbf{spacetime support} of a k-vector field $X$ is
\begin{equation*}
    \supp(X) = \overline{\bigcup_{\varphi \in \sE} \supp(X^{(1)}(\varphi)) \cup \pi_1\supp(X(\varphi))} 
\end{equation*}
where $\pi_1$ is the projection onto the first variable. We denote the set of compactly supported $k$-vector fields on $\sE$ by $\mathcal{X}^k(E)$, or by $\mathcal{X}^k$ when no confusion can arise.
\end{definition}
We note that, due to the anti-symmetry of $k$-vector fields, any other choice of projection gives an equivalent definition. As a word of warning, even though we call these objects \textit{vector fields}, the pairing between vector fields and functionals will in general not be well-defined, as the derivative of a functional is a distributional object and hence will not pair with other distributional objects without proper care.

We spell out the meaning of Proposition \ref{derissmoothfunctional} explicitly in this scenario. Recall that, if $E_a,E_b$ are two vector bundles over $N$ and $P$, then their spaces of sections are nuclear Fréchet spaces. Their completed tensor product is isomorphic to the space of sections of the exterior tensor product bundle $E_a\boxtimes E_b$: 
\begin{equation*}
    \sE(N;E_a) \hatotimes \sE(P;E_b) \cong \sE(N\times P; E_a\boxtimes E_b). 
\end{equation*}
For trivial bundles, this is shown in Theorem 51.6 in \cite{Treves1967}. The general statement then follows from the fact that any vector bundle can be exhibited as a summand of a trivial bundle.

As nuclear Fréchet spaces are Montel, the compact-open topology is equivalent to the bounded-open topology:
\begin{equation*}
L_c(\sE^{\hatotimes n},\sE'(M^k;E^{\boxtimes k})) = L_b(\sE^{\hatotimes n},\sE'(M^k,E^{\boxtimes k})),
\end{equation*}
see e.g.\ Proposition 34.5 in \cite{Treves1967}. Using standard theorems on nuclear Fréchet spaces, see e.g.\ equations 50.16 through 19 in the same reference, we have the following version of the Schwartz kernel theorem:
\begin{dmath*}
L_b(\sE^{\hatotimes n},\sE'(M^k;E^{\boxtimes k})) \cong \left(\sE^{\hatotimes n}\right)^'\hatotimes \ \sE'(M^k;E^{\boxtimes k}),\cong \left(\sE(M^n;E^{\boxtimes n})\ \hatotimes \  \sE(M^k;E^{!\boxtimes k})\right)^',\cong \sE'(M^{n+k};E^{! \boxtimes n}\boxtimes E^{\boxtimes k}).
\end{dmath*}
Hence we can view $X^{(n)}$ as a functional valued in this space of distributional sections. We use this viewpoint to discuss the singular structure of the derivatives of multivector fields.

\subsection{Dynamics}\label{dynamics}
 We consider a theory where the dynamics are given by a linear, Green hyperbolic operator as defined in \cite{Baer2015}: 
\begin{equation*}
    P: \sE \to \sE^!,
\end{equation*}
Morally speaking, we view this operator as the linearisation of the Euler-Langrange equations corresponding to some local action. As we are mainly interested in algebraic properties of the QFT defined by this operator we take it to be fixed, rather than letting it depend on a configuration $\varphi\in\sE$. The case where $P$ is allowed to vary is treated in \cite{Brunetti2019}. 

By definition of a Green hyperbolic operators, there are unique retarded and advanced Green's functions. Following Bär in \cite{Baer2015}, we view them as maps on sections with past or future compact support:
\begin{align*}
    \Delta^R:\sE^!_{\pc} \to \sE^{}_{\pc} \\
    \Delta^A:\sE^!_{\fc} \to \sE^{}_{\fc} \\
\end{align*}
that invert $P$ when restricted to those spaces. Through the Schwartz-kernel Theorem, we can equivalently view $\Delta^{R/A}$ as bidistributions valued in the square of $E$:
\begin{equation*}
    \Delta^{R/A} \in \sD'(M^2, E^{\boxtimes 2})
\end{equation*}

From the advanced and retarded Green's functions, one can then define the \textit{commutator function} to be
\begin{equation*}
    \Delta = \Delta^R - \Delta^A: \sE^!_c \to \sE,
\end{equation*}
which is well defined as $\sE^!_c \subseteq \sE^!_{\pc} \cap \sE^!_{\fc}$. This function has the property that $P \circ \Delta = 0$, i.e.\ it maps compactly supported data to solutions of $P$.

The classic example of a Green hyperbolic operator is that of a normally hyperbolic operator. If $P$ is normally hyperbolic, then it follows from the propagation of singularities theorem that 
\begin{equation} 
\label{normally_hyperbolic}
    \WF(\Delta) = \{(x,y;\xi,\xi') \subset \dot{T}^*M^2 \, | \, (x,\xi) \sim (y,-\xi')\},
\end{equation}
where $(x,\xi) \sim (y,-\xi')$ if the vector $g^{-1}(\xi,\_)$ is tangent to a \textbf{null} geodesic from $x$ to $y$, such that $-\xi'$ is the parallel transport of $\xi$ along this curve. This implies that $\xi$ and $\xi'$ are both null vectors, one future pointing and one past pointing. 

There are other interesting Green hyperbolic operators that we would like our constructions to apply to however. For example, the Proca field is Green hyperbolic but not normally hyperbolic. It turns out that the usual Hadamard condition that one applies to states of the Klein-Gordon theory is not the correct one for this theory, see e.g.\ \cite{Fewster_2003,Moretti_2023}. Another example where more general microlocal constraints are needed when studying birefringence in pre-metric electrodynamics, as in \cite{Fewster_2018A}. Finally, when considering theories induced by an action principle as in \cite{Brunetti2019}, the operator $P$ is allowed to vary with the configuration $\varphi$. It was realised in \cite{Dabrowski2014c} that one needs to restrain the speed of propagation at different configurations both above and below, so that more general cones than the light cone become relevant.

Hence we choose to take a general approach, and impose minimal constraints on the singular structure of the commutator function. Looking at \eqref{normally_hyperbolic}, the only feature that is relevant to us here is the fact that the two parts of the covectors are in `opposite' directions, one along the forward light cone, and one along the backward light cone. We keep this directional feature, but generalise the cones that we use in the definition. 

So let $\mathcal{V}\subset \dot T^* M$ be a closed cone, such that the convex hull of $\mathcal{V}$ does not intersect the zero section of $T^*M$. In particular, this implies that $\mathcal{V}\cap -\mathcal{V} = \emptyset$. We view the cone $\mathcal{V}$ to be part of the background data of the field theory, and take it to be fixed for the remainder of the text. We impose then, in addition to the fact that $P$ is Green hyperbolic, that the wavefront set of the commutator function of $P$ is bounded by
\begin{equation}\label{wavefrontsetcommutatorfunction}
    \WF(\Delta) \subset \left( \mathcal{V} \times -\mathcal{V}\right) \cup \left( -\mathcal{V} \times \mathcal{V}\right).
\end{equation}
The commutator function is too singular for the definition of the $\star$-product of local functionals, as it might not have a well defined square (as can be seen from the wavefront set immediately). To continue our analogy with normally hyperbolic operators, see e.g.\ \cite{Radzikowski1996}, we also assume existence of a (generalised) Hadamard two point function: This is a distributional bisolution to the equations of motion, i.e.\ 
\begin{align*}
    (P\otimes \mathbbm{1})\Delta^+ = 0 \\
    (\mathbbm{1} \otimes P) \Delta^+ =0,
\end{align*}
whose antisymmetric part equals $\frac{i}{2}\Delta$. Furthermore, we impose the \textit{generalized Hadamard condition}:
\begin{equation}\label{wavefronttwoppoint}
    \text{WF}(\Delta^+) \subset \mathcal{V} \times -\mathcal{V}.
\end{equation}
In a sense, this is a frequency splitting of $\Delta$, as we have discarded all covectors along $-\mathcal{V}$ in the first argument. We refer the reader to \cite{Fewster2024} for a discussion of the generalised Hadamard condition and its consequences. Throughout the rest of this text, we assume that an operator $P$ satisfying these conditions is given, and that a choice of two-point function $\Delta^+$ has been made. 

As an aside, we mention that by generalizing the forward light cone to a general cone $\mathcal{V}$, very little of our construction rests on the presence of a spacetime metric at this point. In fact, the only part where it still features is in fixing the causal structure on $M$. In that sense, we could have assumed our background structure to be given by a general manifold $\mathcal{M}$, endowed with a suitable notion of causal structure and a field of cones that bounds the propagation of singularities arising from the Green's functions of $P$. 
\section{The Koszul complex for smooth functionals}\label{KoszulComplex}
Morally speaking, we want to view the configurations $\varphi$ in the kernel of $P$ as the `physical' configurations of the theory. As such, the observables of the theory should be given by functionals on $\textup{Sol} = \ker(P)$. We will denote the set of all compactly supported functionals $\textup{Sol} \to \mathbb{C}$ by $\mathcal{F}_S$.

It is more convenient however, in light of discussing interactions, to take an off-shell approach, i.e.\ to consider functionals $F$ on the whole of $\sE$. We denote the inclusion map of $\textup{Sol}$ into $\sE$ by $\iota$. The pullback by this map defines a restriction map,
\begin{align*}
    \iota^*:\:&\mathcal{F} \to \mathcal{F}_S \\
    &F \mapsto F|_{\textup{Sol}}.
\end{align*}
This map is surjective: any functional on $\textup{Sol}$ can be precomposed with a projection onto $\textup{Sol}$ to obtain a functional on $\sE$. We show below that continuous projections exist. The kernel of $\iota^*$ is the \textit{on-shell ideal}
\begin{equation*}
   \ker(\iota^*) = \{F\in \mathcal{F} \: | \: F|_{\textup{Sol}} = 0 \}.
\end{equation*}

The main goal of this section is to describe a resolution, in the sense of homological algebra, of the map $\iota^*$ in terms of multivector fields. That is, we will extend the map $\iota^*$ to be part of a chain complex
\begin{equation*}\label{Koszul}
    \ldots \xrightarrow{} \mathcal{X}^2 \xrightarrow{\delta} \mathcal{X}^1 \xrightarrow{\delta} \mathcal{F}\xrightarrow{\iota^*} \mathcal{F}_S \xrightarrow{} 0,
\end{equation*}
that is exact. We put $\mathcal{F}$ in degree $0$ here, and $\mathcal{F}_S$ in degree $-1$.

Let $X\in\mathcal{X}^k$, we map it to a $(k-1)$-vector field as follows:
\begin{equation}\label{differential}
    \delta (X) (\varphi) = X(\varphi)\{P\varphi,\_\}.
\end{equation}
Clearly, as $k$-vector fields are antisymmetric, this operator squares to zero in degrees $>1$. Furthermore, if $X \in \mathcal{X}^1$ and $\varphi \in \textup{Sol}$, we have
\begin{equation*}
    \delta X (\varphi) = X(\varphi)\{P\varphi\} =0,
\end{equation*}
so that $\delta \mathcal{X}^1 \subset \ker(\iota^*)$. We shall show that this is an equality, and that the homology of the complex $\left(\mathcal{X}^\bullet,\delta\right)$ vanishes in higher degrees. In the literature of gauge theory quantization, see e.g.\ \cite{Hennaux}, it is also said that `the on-shell ideal is generated by the equations of motion'.

Our main trick is to decompose a general configuration $\varphi$ into a part along Sol, and a complementary part. For this, we first define a projection on the solution space. This projection is not unique as there is no canonical structure, like an inner product, that we can use to define it. We parameterise the arbitrariness of our approach by a function $\theta$, with the property that
\begin{itemize}
\item 
$\{x \in M  \: | \: \theta(x)=0\}$ is future compact, and
\item 
$\{x \in M  \: | \: \theta(x)=1\}$ is past compact.
\end{itemize}
As an example of such a function, we might select two Cauchy surfaces for $M$, one lying strictly in the future of the other, and set $\theta$ equal to zero to the past and equal to one to the future, varying smoothly between the two surfaces. Later, the proof of the time-slice axiom will boil down to making a judicious choice for $\theta$, but for now we allow it to be arbitrary.

We define the following operator:
\begin{align*}
\alpha: \:& \sE^! \to \sE,\\
    &h \mapsto  \Delta^A (1-\theta) h+\Delta^R \theta h, 
\end{align*}
which is well-defined due to the support properties of $\theta$. It is a one-sided inverse to $P$, as it satisfies
\begin{equation*}
    P\alpha h = h.
\end{equation*}
Using this map, we can define a projection onto the solution space:
\begin{equation}\label{gamma0}
\begin{aligned}
    \gamma_0  : \:& \mathcal{E}\to \mathcal{E}, \\
     &\varphi \mapsto \varphi - \alpha P\varphi .
\end{aligned}
\end{equation}
We extend this to a map that takes a configuration $\varphi$ to the line from $\gamma_0\varphi$ to $\varphi$:
\begin{align*}
    \gamma:\:&  [0,1] \times \sE \to \sE, \\
            &(\lambda,\varphi) \mapsto \varphi + (\lambda-1)\alpha P \varphi.
\end{align*}
It is straightforward to check that $\gamma$ is continuous.

If $F$ is a functional on $\sE$, then we can pull back by $\gamma_0$:
\begin{equation*}
    \gamma_0^* F(\varphi) = F(\gamma_0 \varphi) .
\end{equation*}
As $\gamma_0$ is the identity on $\textup{Sol}$, we get that $F-\gamma_0^* F$ vanishes on the solution space. The operator $\gamma_0^*$ changes the support of the functional it acts on in the following way.
\begin{prop}\label{supportsprop}
    Let $B$ be any topological vector space. If $F: \sE \to B$ is a smooth functional, not necessarily compactly supported, then 
\begin{equation}\label{supportpullback}
    \supp(\gamma_0^*F) \subset J \supp(F) \cap \supp(d\theta).
\end{equation}
Similarly, if $\lambda \neq 0$, then
\begin{equation}\label{supportpullbacklambda}
    \supp(\gamma_\lambda^*F) \subset \left (J \supp(F) \cap \supp(d\theta) \right) \cup \supp(F).
\end{equation}
If $F$ is compactly supported, then so is $\gamma_\lambda^* F$ for all $\lambda\in [0,1]$.
\end{prop}
\begin{proof}
Suppose that $x\notin J\supp(F)$. As $J\supp(F)$ is closed, $\supp(F)^\perp=\left(J\supp(F)\right)^c $ is an open neighbourhood of $x$. Suppose that $\psi \in \sE$ is supported in $\supp(F)^\perp$, then it follows from the support properties of $\Delta^{A/R}$ that
\begin{equation*}
    \supp(\gamma_\lambda \psi) \subset J \left(\supp(F)^\perp\right) \subset \supp(F)^c.
\end{equation*}
Hence
    \begin{equation*}\label{psinotinsupport}
       \gamma_\lambda^*F(\varphi+\psi)= F(\gamma_\lambda\varphi+\gamma_\lambda\psi) =\gamma_\lambda^*F(\varphi).
    \end{equation*}
for all $\varphi \in \sE$, so that $x\notin \supp(\gamma_\lambda^* F)$.

Suppose that $x\notin \supp(d\theta)$, then there exists a neighbourhood $U$ of $x$ so that $\theta$ is constant on $U$, with value $t$. If $\psi$ is supported inside $U$, then
\begin{equation}\label{thetaconstant}
    \gamma_\lambda \psi = \psi - (\lambda-1)\left(\Delta^A(1-t)P\psi - \Delta^R t P\psi\right)=\lambda \psi,
\end{equation}
where we have used that $(1-t)P\psi$ is past compact and $tP\psi$ is future compact to cancel the propagators. Therefore
\begin{equation*}
       \gamma_\lambda^*F(\varphi+\psi)= F(\gamma_\lambda\varphi+\lambda\psi)
\end{equation*}
 If $\lambda =0$, this equals $\gamma_0^*F(\varphi)$, meaning that $x\notin \supp(\gamma_0^*F)$. If $\lambda\neq 0$ and additionally $\supp(\psi)\cap \supp(F)=\emptyset$, we get that $x \notin \supp(\gamma_\lambda^*F)$.

Putting all this together, we have shown that
\begin{align*}
    J\supp(F)^c\cup \supp(d\theta)^c &\subset\supp(\gamma_0^*F)^c, \\
    J\supp(F)^c\cup \left (\supp(d\theta)^c\cap \supp(F)^c\right) &\subset \supp(\gamma_\lambda^*F)^c
\end{align*}
and equations \eqref{supportpullback} and \eqref{supportpullbacklambda} follow from taking complements on both sides. 

Finally, if $\supp(F)$ is compact, then $J\supp(F)$ is spatially compact, whereas $\supp(d\theta)$ is temporally compact by the assumptions on $\theta$. As the intersection of a spatially compact set with a temporally compact set is compact, see e.g.\ Lemma 1.9 in \cite{Baer2015}, it follows that $\supp(\gamma_\lambda^*F)$ is compact.
\end{proof}

With these preliminaries out of the way we now define an explicit homotopy operator that will imply that the complex \eqref{Koszul} is exact. In degree $-1$, i.e.\ on $\mathcal{F}_S$, we pull back by the map $\gamma_0$, which we now (and only now) view as a map $\sE \to \textup{Sol}$:
\begin{align*}
    \mathcal{H}_{-1}:\:   &\mathcal{F}_S \to \mathcal{F},\\
                        &F \mapsto F \circ \gamma_0.
\end{align*}
The relation between this map and $\gamma_0^*$ is given by
\begin{equation*}
    \gamma_0^* = \mathcal{H}_{-1} \circ \iota^*.
\end{equation*}
In higher orders, we define maps
\begin{equation*}
    \mathcal{H}_l :\mathcal{X}^l \to \mathcal{X}^{l+1}\\
\end{equation*}
by
\begin{equation}\label{defhomotopyoperator}
    \mathcal{H}_l Y(\varphi)\{h_1,\ldots,h_{l+1}\} = \sum_{i=1}^{l+1} (-1)^{i-1}\int_0^1 Y^{(1)}(\gamma_\lambda \varphi)\{\alpha h_i;h_1, \ldots,\widehat{h_i},\ldots,h_{l+1}\}\lambda^l d\lambda.
\end{equation}
It follows from Propositions \ref{Leibnizintegralrule}, \ref{derissmoothfunctional} and \ref{supportsprop} that $\mathcal{H}_lY \in \sX^{l+1}$. When no confusion can arise, we drop the subscript $l$ and denote all these maps by $\mathcal{H}$. These maps fit into the following diagram:
 \begin{equation}
 \begin{tikzcd}[column sep = large]
         \ldots \ar[r]  &\mathcal{X}^2 \ar[r,"\delta"] \ar[ld,"\mathcal{H}_2"'] &\mathcal{X}^1 \ar[r,"\delta"] \ar[ld,"\mathcal{H}_1"'] &\mathcal{F} \ar[r,"\iota^*"] \ar[ld,"\mathcal{H}_0"'] &\mathcal{F}_S \ar[ld,"\mathcal{H}_{-1}"'] \ar[r] &0 \\
         \ldots \ar[r]  &\mathcal{X}^2 \ar[r,"\delta"] &\mathcal{X}^1 \ar[r,"\delta"] &\mathcal{F} \ar[r,"\iota^*"] &\mathcal{F}_S \ar[r] &0 
 \end{tikzcd}
 \end{equation}
\begin{prop}\label{propostionhomoperator}
The maps $\{\mathcal{H}_l\}_{l=-1}^\infty$ define a chain homotopy between $\mathbbm{1}$ and $0$, that is
\begin{equation}\label{homotopyoperator}
    \mathbbm{1}= \delta \mathcal{H} + \mathcal{H} \delta
\end{equation}
in degrees $> 0$, and
\begin{equation}\label{homotopyoperatordegree0}
  \mathbbm{1}= \delta \mathcal{H} + \gamma_0^*  
\end{equation}
in degree $0$.
\end{prop}

\begin{proof}
Let $X \in \mathcal{X}^k(M)$, where we allow $k=0$ as well. Let $\varphi \in \sE$ and $h_1,\ldots, h_k\in \sE^!$. From the fundamental theorem of calculus applied to the function
\begin{equation*}
   \lambda \mapsto X(\gamma_\lambda \varphi)\{\lambda h_1,\ldots, \lambda h_k\},
\end{equation*}
we note that
\begin{multline}\label{fundth}
   X(\varphi)\{h_1,\ldots,h_k\} - X(\gamma_0\varphi)\{0,\ldots,0\} \\
    =\int_0^1 X^{(1)}(\gamma_\lambda \varphi)\{\alpha P \varphi; h_1,\ldots h_k\}\lambda^k d\lambda + k\int_0^1X(\gamma_\lambda \varphi)\{h_1,\ldots,h_k\} \lambda^{k-1} d\lambda
\end{multline}

In particular, if $k=0$, this is exactly equation \eqref{homotopyoperatordegree0}. Hence we take $k>0$ in what follows. We calculate, for $g\in \sE$:
\begin{equation*}
    \frac{\delta}{\delta \varphi} (\delta X(\varphi)))\{g\} = X^{(1)}(\varphi)\{g;P\varphi, \, \_\,\} + X(\varphi)\{Pg,\,\_\,\}.  
\end{equation*}
Feeding this into $\mathcal{H}$, we find that 
\begin{dmath}\label{Hdelta}
    \mathcal{H}\delta X (\varphi)\{h_1,\ldots h_k\} = \sum_{i=1}^{k} (-1)^{i-1}\int_0^1 (\delta X)^{(1)}(\gamma_\lambda \varphi)\{\alpha h_i;h_1, \ldots,\widehat{h_i},\ldots,h_{k}\}\lambda^{k-1} d\lambda
    =\sum_{i=1}^k (-1)^{i-1}\int_0^1 X^{(1)}(\gamma_\lambda \varphi)\{\alpha h_i;P\gamma_\lambda\varphi,h_1,\ldots,\widehat{h_i},\ldots,h_k\} \lambda^{k-1} d\lambda
    +\sum_{i=1}^k (-1)^{i-1}\int_0^1 X(\gamma_\lambda \varphi)\{P\alpha h_i,h_1, \ldots,\widehat{h_i},\ldots,h_{k}\}\lambda^{k-1}d\lambda
    = \sum_{i=1}^k (-1)^{i-1} \int_0^1 X^{(1)}(\gamma_\lambda \varphi)\{\alpha h_i;P\varphi,h_1,\ldots,\widehat{h_i},\ldots,h_k\} \lambda^{k} d\lambda
    +k\int_0^1 X(\gamma_\lambda \varphi)\{h_1 \ldots,h_{k}\}\lambda^{k-1}d\lambda.
\end{dmath}
In the last step, we have used that 
\begin{align*}
    P\gamma_\lambda \varphi &=  \lambda P\varphi,
\end{align*}
For the other term, we calculate
\begin{dmath}\label{deltaH}
    \delta \mathcal{H}X (\varphi)\{h_1,\ldots,h_k\} = \mathcal{H}X(\varphi)\{P\varphi,h_1,\ldots,h_k\} 
    = \int_0^1 X^{(1)}(\gamma_\lambda \varphi)\{\alpha P\varphi;h_1,\ldots,h_k\}\lambda^k d\lambda
    +\sum_{i=1}^k (-1)^i \int_0^1 X^{(1)}(\gamma_\lambda \varphi)\{\alpha h_i;P\varphi,h_1,\ldots,\widehat{h_i},\ldots,h_k\}\lambda^k d\lambda.
\end{dmath}
At this point the proposition follows from combining equations \eqref{fundth}, \eqref{Hdelta} and \eqref{deltaH}, noting that the alternating factors cancel exactly.
\end{proof}

\begin{corollary}
        The complex $\left(\mathcal{X}_{}^\bullet(E),\delta\right)$ is a resolution of the space of on-shell functionals $\mathcal{F}_S(E)$. 
\end{corollary}
\subsection{Time-slice axiom}
A desirable (perhaps even defining) property of an algebraic quantum field theory is the time-slice axiom. Morally speaking, this axiom states that the observables relating to a spacetime $M$ can be localised arbitrarily close to a Cauchy surface of that spacetime. More concretely, if $N\subset M$ is a region in spacetime, then we can define
\begin{equation*}
    \sX(N) = \{ F \in \sX(M) \: | \: \supp(F) \subset N \}.
\end{equation*}
We have dropped the bundle $E$ from the notation for now. There is an obvious inclusion map $i:\sX(N)\to \sX(M)$. An AQFT on $M$ satisfies the time-slice axiom if $i$ is an isomorphism when $N$ contains a Cauchy surface of $M$. In our present differential graded context, we impose instead that it is a quasi-isomorphism. The resolution from the previous section can be used to prove the time-slice axiom for smooth compactly supported functionals.
\begin{theorem}\label{quasiinverse}
If $N\subset M$ is a region containing a Cauchy surface for $M$, then the inclusion map $i:\sX(N) \to \sX(M)$ is a quasi-isomorphism. 
\end{theorem}
\begin{proof}
The proof is an adaptation of a standard argument to the functional formalism. Fix two Cauchy surfaces $\Sigma^+$ and $\Sigma^-$, both contained in $N$, such that $\Sigma^+ \subset I^+(\Sigma^-)$. Take a function $\theta \in \mathcal{E}(M)$ which is identically 1 to the future of $\Sigma^+$ and identically 0 to the past of $\Sigma^-$. This means in particular that $d\theta = 0$ outside of $N$.

We claim that the maps
\begin{align*}
    &\Theta_0 = \gamma_0^* \\
    &\Theta_k = 0 , \ k>0            
\end{align*}
form a quasi-inverse to $i$, where $\gamma_0$ is defined with respect to the $\theta$ just chosen.
It follows from Lemma \ref{supportsprop} that 
\begin{equation*}
    \supp(\gamma_0^*F) \subset J \supp(F) \cap \supp(d\theta) \subset N,
\end{equation*}
so that 
\begin{equation*}
    \gamma_0^* : \sF(M)\to \sF(N).
\end{equation*}
Obviously $0 \in \sX^k$ is supported in $N$ (as it has empty support), so that $\Theta: \sX(M)\to \sX(N)$. We note that $\gamma_0^* \circ \delta =0$ so that $\Theta$ is a chain map. Furthermore, in Proposition \ref{propostionhomoperator}, we showed that 
\begin{equation*}
    F-\gamma_0^*F = \delta \mathcal{H} F
\end{equation*}
 so that $F$ and $\gamma_0^*F$ are equivalent in homology. As the homology of the complex of multivector fields is trivial in higher degrees, there is nothing more to check.
\end{proof}

\section{Undesirable features of microcausal functionals}\label{pathologicalfeatures}
So far we have considered the space of \textbf{all} smooth functionals on $\sE$ to describe a physical system. However, this class is too large for the purposes of field theory. In defining algebraic structures, like the Poisson bracket or the $\star$-product, we pair the derivatives of two functionals with a distribution. For example, the Poisson bracket of a classical field theory is formally given by
\begin{equation}\label{PoissonBracket}
    \{F,G\}(\varphi) = \Delta\left(F^{(1)}(\varphi) \otimes G^{(1)}(\varphi)\right),
\end{equation}
which is in general not well-defined as we are matching distributional indices in this prescription.

It is at this point that microlocal analysis is introduced into the functional formalism, by restricting to the \textit{microcausal} functionals. These are functionals with a prescribed singular structure for their derivatives. This class is small enough to allow the Poisson bracket and the $\star$-product to be well-defined, but at the same time large enough to contain all functionals that are of interest in perturbative treatments of QFT, i.e.\ the multilocal functionals. We employ a slight generalisation to the usual concept, to allow for  using the general cone $\mathcal{V}$ we introduced in Section \ref{dynamics}.\footnote{This makes the name microcausal somewhat misleading, as the cone $\mathcal{V}$ need no longer be tied to the causal structure of $M$. One recovers the usual definition when $\mathcal{V}$ is the forward light cone.} 
\begin{definition}
We define a sequence of cones $\Gamma_n \subset \dot T^*(M^n)$ by
\begin{equation*}
    \Gamma_n = \mathcal{V}^{\dot\times n} \cup (-\mathcal{V})^{\dot\times n}.
\end{equation*}
A $k$-vector field $X$ is microcausal if
\begin{equation*}
    \WF\left(X^{(n)}(\varphi)\right) \cap \Gamma_{n+k} = \emptyset  \: \forall \: \varphi \in \mathcal{E}.
\end{equation*}

We denote the set of compactly supported microcausal $k$-vector fields by $\mathcal{X}^k_{\mu c}(E)\subset \mathcal{X}^k(E)$, and the set of compactly supported microcausal functionals by $\mathcal{F}_{\mu c}(E)= \sX^0_{\mu c}(E)$. When no confusion can arise, we will denote them by $\sF_{\mu c}$ and $\sX^k_{\mu c}$.
\end{definition}
A different, more standard, way of phrasing this is to say that $X$ is microcausal iff
\begin{equation*}
    X^{(n)}(\varphi) \in \egc{n+k}(M^{n+k};E^{! \boxtimes n}\boxtimes E^{\boxtimes k}) \: \forall \: n\in\mathbb{N}, \varphi \in \sE,
\end{equation*}
noting that $\Gamma_m = - \Gamma_m$ for all $m$. However, we find it more convenient to view $X^{(n)}(\varphi)$ as a map that extends to $\dg{n+k}$, which is why we favour this definition. We note also that some authors include compact support in the definition of microcausal functionals, but we choose to disentangle the singular behaviour and the support properties so as to be more in line with \cite{Brouder2018}. 
Because of the explicit form of the wavefront set of $\Delta$, given in equation \eqref{wavefrontsetcommutatorfunction}, the pairing in the Poisson bracket in equation \eqref{PoissonBracket} is well-defined for microcausal functionals. 

Even though this definition is standard in the literature, there are several problematic features relating to this class that have thus far gone unnoticed. The first is that the homotopy operator $\mathcal{H}$ from the previous section does not respect the subcomplex of microcausal multivector fields. The problem is that, for $X\in\sX^k_{\mu c}$, an integral of the form
 \begin{equation*}
     \int_0^1 X^{(n+1)}(\gamma_\lambda \varphi) d \lambda
 \end{equation*}
 does not respect wavefront structure in general, as the integrand is not guaranteed to be bounded in $\egc{n+1+k}$ as $\lambda$ varies. It may therefore happen that unwanted singularities pop up when performing the integration in equation \eqref{defhomotopyoperator}. Ultimately, this is an obstruction to proving the time-slice axiom for microcausal functionals using the customary method. A detailed counterexample to this fact will appear in \cite{Thesis}.

The second, more serious, problematic feature is that the Poisson-bracket of two microcausal functionals is not smooth in general, or even continuous. When attempting to calculate the derivative of $\{F,G\}$, it is tempting to use the Leibniz rule to conclude that
\begin{equation}\label{derivativeofbracket}
    \{F,G\}^{(1)}(\varphi)\{h\} \overset{?}{=} \Delta\left(F^{(2)}(\varphi)\{h,\_\}\otimes G^{(1)}(\varphi)\right) + \Delta\left(F^{(1)}(\varphi),G^{(2)}(\varphi)\{h,\_\}\right)
\end{equation}
The right-hand side of this equation is well defined, and is given by pairing $h$ with an element of $\egc{1}$, due to the fact that $F$ and $G$ are microcausal, see e.g.\ Section 3 \cite{Brunetti2019}. However, despite the fact that there exists an obvious candidate for the derivative, it is not guaranteed that it matches
\begin{equation*}
    \lim_{t\to 0} \frac{1}{t}\left( \{F,G\}(\varphi +t h)-\{F,G\}(\varphi )\right),
\end{equation*}
or indeed that the limit even exists. While $F^{(1)}$ and $G^{(1)}$ take values in the subspace $\egc{1}(M)\subset \sE'(M)$, they are not assumed to be smooth with respect to the  topology on that space. Whilst $\{F,G\}$ is well defined as a mapping $\sE \to \mathbb{C}$, it is \textit{a priori} not clear whether $\{F,G\}$ is even continuous. This turns out to not be true in general, and in Propositions~\ref{Regular} and \ref{Counterexample} we give an explicit counterexample. For concreteness, we choose to take a trivial bundle $E$, so that $\sE = \sE(M)$ for the remainder of this section. 

A \textbf{regular} functional is a smooth, compactly supported functional $F:\sE \to \mathbb{C}$, such that $F^{(n)}(\varphi) \in \mathcal{D}(M^n)$ for all $ \varphi \in \sE$. In particular, regular functionals are microcausal (whether one includes compact support in that definition or not). Suppose further that we have an element $W \in \sD'(M\times M)$. We define a paring of regular functionals as follows:
\begin{equation}\label{Wbracket}
    \{F,G\}_W (\varphi) = W(F^{(1)}(\varphi) \otimes G^{(1)}(\varphi)),
\end{equation}
which is well-defined as both $F^{(1)}$ and $G^{(1)}$ are elements of $\sD(M)$. We first simplify the problem by rephrasing it in terms of maps with finite dimensional domain.
\begin{definition}\label{Regular map}
A smooth map $A:\mathbb{R}^n \to \sE'(M)$ is \textbf{regular} if for any multi-index $\alpha\in\mathbb{N}^n$ and $\xi\in\mathbb{R}^n$,
\begin{equation*}
     \partial^\alpha A(\xi)\in \sD(M),
\end{equation*}
and 
$\bigcup_{\xi\in \mathbb{R}^n} \supp A(\xi)$ is compact.
\end{definition}

\begin{prop}
\label{Regular}
    If the functional $\{F,G\}_W$ is smooth for all regular functionals $F$ and $G$, then the map 
    \begin{equation*}
        W(A\otimes B) : (\xi,\zeta) \mapsto W(A(\xi)\otimes B(\zeta))
    \end{equation*}
     is smooth for all regular maps $A$ and $B$.
\end{prop}

\begin{proof}
    Let $A:\mathbb{R}^n\to \sE'(M)$ and $B:\mathbb{R}^m\to \sE'(M)$ be regular maps. The proof proceeds by constructing two regular functionals that `lift' $A$ and $B$ to $\sE$ in an appropriate fashion. We choose $\varphi_1,\ldots,\varphi_n, \tilde{\varphi}_1,\ldots, \tilde{\varphi}_m\in \sE$ linearly independent, satisfying
    \begin{equation*}
        \supp(\varphi_i)\cap\left(\bigcup_{\xi\in\mathbb{R}^n} \supp(A(\xi)) \cup \bigcup_{\zeta\in\mathbb{R}^m} \supp(B(\zeta))\right) = \emptyset,
    \end{equation*}
and similar for $\tilde{\varphi_j}$, which is possible by the assumption of compact support on $A$ and $B$. This implies in particular that
\begin{equation}\label{phisinkernel}
    \partial^\alpha A(\xi)\{\varphi_i\}=\partial^\alpha A(\xi)\{\tilde{\varphi}_i\}=0,
\end{equation}
for any multi-index $\alpha\in\mathbb{N}^n$, and similarly for $B$. Furthermore, we choose $\beta_i,\tilde{\beta}_j \in \sD(M)$ for $i=1,\ldots, n$ and $j=1,\ldots,m$, normalized such that
\begin{align*}
    &\beta_i(\varphi_j) =\delta_{ij},\\
    &\tilde{\beta}_i(\tilde{\varphi_j})=\delta_{ij}, \\
    &\beta_i(\tilde{\varphi}_j)= 0 = \tilde{\beta}_i(\varphi_j),
\end{align*}
where we view the $\beta$'s as distributions by a slight bit of abusive notation. This is possible through a Gram–Schmidt orthonormalisation procedure. This defines a continuous linear map
\begin{align*}
    \Vec{\beta}:\:    &\sE \to \mathbb{R}^n, \\
                    &\varphi \mapsto (\beta_i(\varphi))_{i=1}^n.
\end{align*}
We define $\vec{\tilde{\beta}}$ similarly. Finally, we define the map
\begin{align*}
    \Phi:\:   &\mathbb{R}^n \times \mathbb{R}^m \to \sE, \\
            &(\xi,\zeta) \mapsto \sum_{i=1}^n \xi_i\varphi_i + \sum_{j=1}^m \zeta_i\tilde{\varphi}_j,
\end{align*}
which is smooth, and we note that 
\begin{align*}
    \vec{\beta}\Phi (\xi,\zeta) = \xi, \\
    \vec{\tilde\beta}\Phi (\xi,\zeta) = \zeta,
\end{align*}
from the normalisation of the $\beta_i$ and $\tilde\beta_j$.

Now set
\begin{align*}\label{definition}
    F(\varphi) = A(\Vec{\beta}\varphi)\{\varphi\},\\ 
   G(\varphi) = B(\Vec{\tilde{\beta}}\varphi)\{\varphi\}.
\end{align*}
Both $F$ and $G$ are smooth functionals by the chain rule. We show that $F$ is regular, the case for $G$ being similar. The first derivative of $F$ is
\begin{dmath}\label{derofF}
    F^{(1)}(\varphi)\{h\} = 
    A(\vec{\beta}\varphi)\{h\} +
    \sum_{i=1}^m \left(\beta_i(h)\partial_i A\right)(\vec{\beta}\varphi)\{\varphi\}.
    \end{dmath}
From this expression, we see that $F^{(1)}(\varphi)\in \mathcal{D}(M)$, as $h$ gets smeared against either $\beta_i$ or $A(\vec{\beta}\varphi)$, both of which are smooth functions. At higher order, an induction argument shows that 
\begin{equation*}
    F^{(k)}(\varphi)\{h^{\otimes k}\} = \left[\left(\sum_{i=1}^n \beta_i (h) \partial_i\right)^{ k} A\right](\vec\beta \varphi)\{\varphi\} + (k-1)\left[\left(\sum_{i=1}^n \beta_i (h) \partial_i\right)^{k-1} A\right](\vec\beta \varphi)\{ h\}.
\end{equation*}
Here $h$ gets smeared against either $\beta_i$ or some partial derivative of $A$, both of which are smooth functions. Hence $F$ is a regular functional. 

Now, by hypothesis, $\{F,G\}_W$ is a smooth functional. Furthermore, we note that
\begin{equation*}
    F^{(1)}(\Phi (\xi,\zeta))\{h\} = A(\xi)\{h\}+ \sum_{i=1}^m \beta_i(h)\partial_i A(\xi)\left\{\Phi(\xi,\zeta)\right\} = A(\xi)\{h\}
\end{equation*}
by equation \eqref{phisinkernel}, so that
\begin{align*}
    F^{(1)} \circ \Phi (\xi,\zeta) = A(\xi),\\
    G^{(1)} \circ \Phi (\xi,\zeta) = B(\zeta).
\end{align*}
Putting things together, we have shown that $W(A\otimes B) = \{F,G\}_W \circ \Phi$, which is smooth by the chain rule.
\end{proof}

\begin{lemma}
\label{IOFT}
If $f\in\sD(\mathbb{R}^n)$,  $K\in\mathbb{N}$, and $|\:\cdot\:|$ is the euclidean norm, then
\begin{equation*}
\tilde f(\xi) = 
\begin{cases}
    \lvert\xi\rvert^{-K}\hat f(\xi/|\xi|^2) & \xi\neq0\\
    0 & \xi=0
\end{cases}
\end{equation*}
defines a smooth function, $\tilde f\in\sE(\mathbb{R}^n)$.
\end{lemma}
\begin{proof}
First, recall that the Fourier transform $\hat f$ is smooth and falls off faster than any power at $\infty$, so $\tilde f$ is clearly continuous over $\mathbb R^n\smallsetminus 0$. At $0$,
\begin{align*}
\lim_{\xi\to0}\tilde f(\xi) = \lim_{\zeta\to\infty} \lvert\zeta\rvert^K \hat f(\zeta) = 0 = \tilde f(0) ,
\end{align*}
so $\tilde f$ is continuous everywhere.

Denote $\zeta := \xi/\lvert\xi\rvert^2$, so that $\tilde f(\xi) = \lvert\xi\rvert^{-K}\hat f(\zeta)$. Differentiating,
\begin{equation*}
\frac{\partial}{\partial\xi^i}\lvert\xi\rvert^{-K} = -K \lvert\xi\rvert^{-K-2}\xi_i = -K \lvert\xi\rvert^{-K}\zeta_i ,
\end{equation*}
and
\begin{equation*}
\frac{\partial \zeta_j}{\partial \xi_i} = -2\lvert\xi\rvert^{-2}\zeta_j\xi_i + \lvert\xi\rvert^{-2} \delta_{ij} = \lvert\zeta\rvert^2 \delta_{ij} - 2\zeta_i\zeta_j .
\end{equation*}
By the product rule and chain rule, 
\begin{align*}
\frac{\partial}{\partial\xi_i}\tilde f(\xi) &= -K \lvert\xi\rvert^{-K}\zeta_i \hat f(\zeta) + \lvert\xi\rvert^{-K} \left(\lvert\zeta\rvert^2 \delta_{ij} - 2\zeta_i\zeta_j\right)\frac{\partial}{\partial\zeta_j}\hat f(\zeta)
\\ &=\tilde{g_i}(\xi) ,
\end{align*}
where
\begin{align*}
g_i(x) &= iK \nabla_i f(x)  + i \nabla^2\bigl(x_i f(x)\bigr) - 2i \nabla_i\nabla_j\left(x_j f(x)\right)  ,
\end{align*}
since the Fourier transform effectively replaces $\zeta_i \to -i\nabla_i$ and $\frac{\partial}{\partial \zeta_i}\to -i x_i$.
Differentiation and multiplication by $x_i$ preserve smoothness and compact support, therefore $g_i\in\sD(\mathbb{R}^n)$.

By induction, this shows that \emph{any} derivative (of any order) of $\tilde f$ is equal to $\tilde h$ for some $h\in\sD(\mathbb R^n)$ and is therefore continuous. Therefore, $\tilde f$ is smooth.
\end{proof}

Despite the seemingly innocuous definition of a regular map, these objects can be quite badly behaved, as the following proposition shows.

\begin{prop}
\label{Counterexample}
    If $\WF(W)\neq 0$ and $n=\dim M$, then there exists a pair of regular maps 
    \begin{equation*}
        A,B:\mathbb{R}^n\to\sD(M)\subset\sE'(M)
    \end{equation*}
    such that $(\xi,\zeta) \mapsto W(A(\xi)\otimes B(\zeta))$ is not continuous.
\end{prop}
\begin{proof}
Let $(x_1,x_2;\eta_1,\eta_2) \in \WF(W)$, and select coordinate charts defined on neighbourhoods $U_i$ of the $x_i$. We recall that the wavefront set on manifolds is defined with respect to an arbitrary choice of coordinates, and that all choices are equivalent. Hence we may identify the $U_i$ with open subsets of $\mathbb{R}^n$, and $T_{x_i}^*M$ with $\mathbb{R}^n$. In order not to overburden the notation, we make these identifications implicitly throughout this proof.

The integral kernel for the Fourier transform on $U_1$ is $e_{\xi}:\mathbb{R}^n\to\mathbb{C}$, defined for  $\xi \in \mathbb{R}^n$ by
\begin{equation*}
    e_{\xi}(x) = \exp(- i\langle \xi,x\rangle).
\end{equation*}
As $(x_1,x_2;\eta_1,\eta_2) \in \WF(W|_{U_1\times U_2})$,
 there are test functions $\chi_i\in D(U_i)$  such that $\chi_i(x_i) \neq 0$, and such that
 \begin{align*}
     \Psi:\:& \mathbb{R}^{2n}\to\mathbb{C},\\
     & (\xi,\zeta) \mapsto W(\chi_1 e_{\xi}\otimes\chi_2 e_{\zeta})
 \end{align*}
 is not of rapid decay on \textbf{any} conic neighbourhood of $(\eta_1,\eta_2)$, where we view $\chi_1 e_{\xi}$ and $\chi_2e_{\zeta}$ as functions on the whole of $M$ by extending by $0$ outside of $U_i$. 
 
 In particular, $\Psi$ is not of rapid decay on $\mathbb{R}^{2n}$, so that there is an $N\in \mathbb{N}$ such that $|(\xi,\zeta)|^N \Psi(\xi,\zeta)$ is unbounded. This in turn implies, by the binomial theorem, that there are $K,L \in \mathbb{N}$ such that $|\xi|^K |\zeta|^L \Psi(\xi,\zeta)$ is unbounded. 

In order to make this unbounded function into a discontinuous function, we turn $\mathbb{R}^n$ inside out and define
\begin{equation*}
A(\xi) = \begin{cases}
      |\xi|^{-K}\chi_1 e_{\xi/|\xi|^2}\: &\xi \neq 0,\\
     0 & \xi= 0,
    \end{cases} \in \sD(M)
\end{equation*}
and
\begin{equation*}
B(\zeta) = \begin{cases}
      |\zeta|^{-L}\chi_2 e_{\zeta/|\zeta|^2}\: &\zeta \neq 0,\\
     0 & \zeta= 0.
    \end{cases}\in \sD(M)
\end{equation*}
The map $(\xi,\zeta) \mapsto W(A(\xi),B(\zeta))$ is unbounded on any neighbourhood of $(0,0)$, and hence cannot be continuous. Now we just need to show that $A$ and $B$ are regular (in the sense of Def.~\ref{Regular map}). We do this for $A$, as the case for $B$ is identical.

Note that for any $\xi\in\mathbb R^n$, $\supp A(\xi) \subseteq \supp \chi_1 \subset U_1$, so the compact support condition is satisfied and we can work within $U_1$. We just need to show that $A$ is smooth as a map into $\sE'(U_1)$.
 It suffices to check this in the weak topology, as sequences in $\sE'(U_1)$ converge weakly if and only if they converge strongly, see e.g.\ Proposition 34.6 of \cite{Treves1967}. For any $f\in\sE(U_1)$ and  $\xi \neq 0$, we calculate, by the definition of Fourier transform,
\begin{align*}\label{uonfunction}
    A(\xi)\{f\} &= 
\begin{cases}
|\xi|^{-K} \widehat{\chi_1 f}\left({\xi}/{|\xi|^2}\right) & \xi\neq0 \\
0 & \xi=0
\end{cases}\\
&= \widetilde{\chi_1f}(\xi)
\end{align*}
in the notation of Lemma~\ref{IOFT}. By Lemma~\ref{IOFT}, this is smooth for any $f\in\sE(M)$, which means that $A$ is weakly smooth and implies that $A$ is strongly smooth. Therefore, $A$ is a regular map.
\end{proof}
 We arrive at the following theorem.
\begin{theorem}
    The bracket $\{F,G\}_W$ is a smooth functional for all regular functionals $F$ and $G$ if and only if $\WF(W)=\emptyset.$
\end{theorem}
\begin{proof}
    We have already shown that, if $\WF(W)\neq\emptyset$, then there exist regular functionals $F$ and $G$ such that $\{F,G\}_W$ is not smooth. Conversely, suppose that $\WF(W)=\emptyset$, which means that $W\in \sE(M^2)$. Hence it can be viewed as a continuous map
    \begin{equation*}
        \sE'(M)\times \sE'(M) \xrightarrow{\otimes} \sE'(M^2) \xrightarrow{W} \mathbb{C},
    \end{equation*}
    recalling that the tensor product is continuous on $\sE'$, see e.g.\ Theorem 41.1 of \cite{Treves1967}. As $F^{(1)}$ and $G^{(1)}$ are smooth functionals into $\sE'(M)$ by Proposition \ref{derissmoothfunctional}, the chain rule implies that $\{F,G\}_W$ is a smooth functional.
\end{proof}

\section{Equicausal multivector fields}\label{equicausal} 
The rest of this paper is devoted to describing a subclass of the microcausal multivector fields that do allow us to perform these operations. Broadly speaking, we need some smoothness condition for the maps 
\begin{equation}
    F^{(n)}: \sE \to \egc{n}
\end{equation}
to ensure that the Leibniz rule holds for the Poisson bracket and $\star$-product. Furthermore, we need these derivatives to be `integrable' along curves $\gamma$ in $\sE$:
\begin{equation}\label{integralreq}
    \int_a^b F^{(n)}\circ \gamma (t) dt \in \egc{n}
\end{equation}
to ensure that the homotopy operator in equation \eqref{defhomotopyoperator} does not spawn additional singularities. A class of functionals that achieves both these goals is the following:
\begin{definition}\label{equicausaldefinition}
    A $k$-vector field $X$ is \textbf{equicausal} if it is microcausal and, whenever $C\subset \sE$ is a compact set, the set of linear mappings
    \begin{equation*}
        X^{(n)}(C): \dg{n+k}(E^{\boxtimes n} \boxtimes E^{!\boxtimes k}) \to \mathbb{C}
    \end{equation*}
    is equicontinuous, as defined in Definition \ref{equicontinuous}. We denote the complex of compactly supported equicausal multivector fields on $\sE$ by $\sX^\bullet_\ec(E)$, or just by $\sX^\bullet_\ec$ when no confusion can arise. As before, we write also $\sF_\ec=\sX^0_\ec$.
\end{definition}
We unpack the definition somewhat for the case of functionals. If $F$ is a smooth functional, $\varphi \in \sE$ and $n \in \mathbb{N}$, then  $F^{(n)}(\varphi)$ is a continuous map $\sE^n \to \mathbb{C}$. To say that $F$ is microcausal means that there exists a continuous extension $\tilde{F}^{(n)}(\varphi)$
\begin{equation*}
\begin{tikzcd}[column sep = large, row sep = small]
    \sE^n \ar[dr,bend left,"F^{(n)}(\varphi)"] \ar[dd,"\iota"]& \\
    & \mathbb{C}\\
    \dg{n} \ar[ur,bend right,"\tilde{F}^{(n)}(\varphi)"']&
\end{tikzcd}
\end{equation*}
which is necessarily unique as $\sE^n$ is dense in $\dg{n}$. To say that $F$ is equicausal means that, whenever $C\subset \sE$ is a compact set, $\Tilde{F}^{(n)}(C)$ is an equicontinuous set.

The attentive reader will notice that we have not assumed any additional smoothness conditions on $F^{(n)}$. This is because the local boundedness of the derivatives is in fact enough to prove that $F^{(n)}$ is conveniently smooth into $\egc{n}$, which will be strong enough for our purposes. The proof of this fact rests on two technical lemmas about curves of distributions valued in spaces of the form $\el{}$. We give their proofs in Appendix \ref{technicallemmata}.

\begin{restatable}{lemma}{Curveiscontinuous}\label{Curveiscontinuous}
Let $\tilde{E}\to N$ be a vector bundle and $\Lambda \subset \dot{T}^*N$ an open cone. If
\begin{equation*}
    u: \mathbb{R} \to \sE'(N;\tilde{E})
\end{equation*} 
is a smooth curve that maps bounded intervals to equicontinuous subsets of $\el{}(N;\tilde{E})$, and $Z\in \sD'_{-\Lambda^c}(N;\tilde{E}^!)$, then the map
\begin{equation*}
    t \mapsto u_t(Z)
\end{equation*}
is continuous $\mathbb{R}\to\mathbb{C}$. Furthermore, if $a\leq b$, then $\int_a^b u_t dt \in \el{}(N;\tilde{E})$, and
\begin{equation}
    \left(\int_a^b u_t dt\right)(Z) = \int_a^b u_t (Z)dt.
\end{equation}
\end{restatable}
We note here that this lemma implies in particular that \eqref{integralreq} is satisfied for equicausal functionals.

\begin{restatable}{lemma}{CurveisSmooth}\label{CurveisSmooth}
In the same situation as in Lemma~\ref{Curveiscontinuous}, if 
\begin{equation*}
    u: \mathbb{R} \to \sE'(N;\tilde{E})
\end{equation*} 
 is a smooth curve with the property that, for all $n\in \mathbb{N}$, the curve $\partial_t^n u$ maps bounded intervals to equicontinuous subsets of $\el{}(N;\tilde{E})$, then $u$ is smooth as a curve in $\el{}(N;\tilde{E})$.
\end{restatable}
\begin{theorem}\label{ecimpliesconvenient}
If $X\in\sX^k_\ec(E)$, then $X^{(n)}$ is a conveniently smooth as a functional 
\begin{equation*}
    \sE \to \egc{n+k}(M^{n+k};E^{!\boxtimes n}\boxtimes E^{\boxtimes{k}}).
\end{equation*}    
\end{theorem}
\begin{proof}
    Let $\gamma:\mathbb{R} \to \sE$ be a smooth curve. We show that $X^{(n)}\circ \gamma$ is a curve as in the previous lemma, with $\Lambda = \Gamma_{n+k}^c$ and $\tilde{E}=E^{!\boxtimes n}\boxtimes E^{\boxtimes{k}}$. We calculate, using Faá di Bruno's formula:
    \begin{equation*}
        \left(X^{(n)}\circ \gamma\right)^{(m)}(t) = \sum_{\pi\in P_m} X^{(n+|\pi|)}(\gamma_t)\left\{\bigotimes_{I\in\pi} \gamma^{(I)}_t \otimes \_ \right\} 
    \end{equation*}
where the sum runs over all partitions $\pi$ of the set $(1,\ldots,m)$. We can treat each of the terms separately, so that we are looking at a curve of the form
\begin{equation*}
    \Phi_t: Z \mapsto X^{(n+|\pi|)}(\gamma_t)\left\{\bigotimes_{I\in\pi} \gamma^{(|I|)}_t \otimes Z \right\}
\end{equation*}
for some fixed partition $\pi$. Let $(a,b)\subset \mathbb{R}$ be a bounded interval. We view $\Phi_t$ as a composite
\begin{equation*}
    \Phi_t: \dg{n+k} \xrightarrow{\left(\bigotimes_{I\in\pi} \gamma^{(|I|)}_t\right)\otimes \_} \dg{n+k+|\pi|}\xrightarrow{X^{(n+|\pi|)}(\gamma_t)} \mathbb{C}.
\end{equation*}
Both of these factors form equicontinuous sets as $t$ ranges across $(a,b)$. For the first, this follows from the fact that the tensor product is a hypocontinuous map
\begin{equation*}
    \sE^{|\pi|}\times \dg{n+k} \to \dg{n+k+|\pi|},
\end{equation*}
as $\left\{\bigotimes_{I\in\pi} \gamma^{(|I|)}_t\right\}_{t\in(a,b)}\subset \sE^{|\pi|}$ is bounded, and
\begin{equation}
    \emptyset \dot{\times} \Gamma_{n+k} \subset \Gamma_{n+k+|\pi|}.
\end{equation} For the second, this is a consequence of the fact that $X$ is equicausal. The theorem then follows because the composition of equicontinuous sets is equicontinuous.
\end{proof}
\subsection{Examples of equicausal functionals}
We give some examples of equicausal functionals, our main result being that both local functionals in the sense of \cite{Brouder2018}, as well as Wick polynomials, are equicausal. For notational convenience, we shall again take $E$ to be trivial in this section, but all statements admit straightforward extensions to the scenario with nontrivial bundles. We shall prove first an explicit criterion for equicausality.
\begin{prop}\label{checkec}
Suppose that $F$ is a smooth functional on $\sE$ satisfying the property that for any $n\in \mathbb{N}$ and $\varphi \in \sE$, there is a neighbourhood $V$ of $\varphi$, a closed cone $\Xi \subset \Gamma_n^c$ and a compact set $K\subset M^n$ such that $F^{(n)}$ restricts to a continuous map
\begin{equation*}
    \left.F^{(n)}\right|_V : V \to \sD'_{\Xi}(K).
\end{equation*}
Then $F$ is equicausal.
\end{prop}
\begin{proof}
    Let $C\subset \sE$ be a compact set and let $n\in\mathbb{N}$. For each $\varphi\in C$, pick a neighbourhood as in the proposition, and denote it by $V_\varphi$. By restricing the $V_\varphi$ if necessary, we may assume they are all closed neighbourhoods. These $V_\varphi$ cover $C$, and hence we may choose a finite subset $\{\varphi\}_{i=1}^m \subset C$, such that the sets $V_i = V_{\varphi_i}$ still cover $C$. 
    
    By assumption, there exist $\Xi_{i}$ and $K_i$ so that 
    \begin{equation*}
        \left.F^{(n)}\right|_{V_i} : V_{i} \to \sD'_{\Xi_{i}}(K_i)
    \end{equation*}
is continuous. As $C\cap V_i$ is compact, being the intersection of a compact set with a closed set, $F^{(n)}(C\cap V_i)$ is compact in $\sD'_{\Xi_{i}}(K_i)$, and hence bounded. Therefore it is an equicontinuous subset of $\egc{n}$ by Proposition \ref{characterizationequicontinuous}. It follows that
\begin{equation*}
    F^{(n)}(C) = \bigcup_{i=1}^k F^{(n)} (C \cap V_i) \subset \egc{n}.
\end{equation*}
is equicontinuous as well. 
\end{proof}

\subsubsection{Local Functionals}
Morally speaking, a local functionals is a functional that can be written as an integral over spacetime of some function on the jet bundle:
\begin{equation*}
    F(\varphi) = \int f(j_x^k\varphi)dx ,
\end{equation*}
for some $k\in \mathbb{N}$, where $j_x^k\varphi$ is the $k$'th jet-prolongation of $\varphi$ and $f$ is a smooth function on the jet bundle. In a recent paper by Brouder et al.\ \cite{Brouder2018} an alternative characterization of local functionals was given in Theorem VI.3, which is the version that we will use:
\begin{definition}
    A smooth functional $F$ on $\sE$ is \textbf{local} if it satisfies the following three properties:
    \begin{itemize}
        \item $F$ is \textit{additive}, in the sense that
        \begin{equation*}
            F(\varphi_1+\varphi_2+\varphi_3) = F(\varphi_1+\varphi_2)+F(\varphi_2+\varphi_3) - F(\varphi_2),
        \end{equation*}
        whenever $\textup{supp}(\varphi_1)\cap\textup{supp}(\varphi_3)=\emptyset$.
        \item $F^{(1)}(\varphi) \in \sD(M)$ for any $\varphi\in\sE$, and furthermore,
        \item $F^{(1)}$ is implemented by a \textbf{smooth} map
\begin{equation*}
    \nabla F: \mathcal{E} \to \mathcal{D},
\end{equation*}
\end{itemize}
that is to say, for all $h\in\sE$
\begin{equation}\label{implementderivative}
    F^{(1)}(\varphi)\{h\} = \int \nabla F(\varphi)(x)\,h(x)\,dx. 
\end{equation}
\end{definition}
We shall need a generalisation of Lemma VI.9 in \cite{Brouder2018} that gives a concrete expression for the $n$'th derivative of a local functional. The proof of this fact is very similar to the proof of the $n=2$ case in the original paper. We will therefore sketch the proof, and defer the detailed proof to \cite{Thesis}.
 
\begin{prop}\label{localfuncform}
If $F$ is a local functional, $l\in\mathbb{N}$ and $\varphi_0 \in \sE$, then there is a neighbourhood $V$ of $\varphi_0$ and a $k\in \mathbb{N}$ such that in any coordinate system of $M$, the distributional kernel of $F^{(l)}$ takes the form
\begin{equation}\label{explicitform}
    F^{(l)}(\varphi)(x_1,\ldots,x_l) = \sum_{|\Vec{\alpha}| \leq k} f_{\Vec{\alpha}}(\varphi)(x_1)\partial^{\Vec{\alpha}}\left[\delta(x_2-x_1)\ldots \delta(x_l-x_1)\right],
\end{equation}
where the $f_{\Vec{\alpha}}$ are smooth maps $V \to \sD (K)$ for some compact $K\subset M$.
\end{prop}

\begin{proof}
As in Proposition III.11 of \cite{Brouder2018}, we pick a neighbourhood of $\varphi_0$ such that $F|_V$ is supported within some compact $K$, and of order bounded by some $k\in\mathbb{N}$. By going to a coordinate system and using a partition of unity if necessary, we may assume that $K\subset \mathbb{R}^n$ where $n$ is the dimension of $M$. By the additive property of $F$, it follows from Proposition V.5 in \cite{Brouder2018} that the $n$'th derivative is supported on the thin diagonal of $M$, i.e.\ on the set
\begin{equation*}
    \Delta_n=\{(x,x,\ldots,x) \in M^n \: | \: x \in M\} . 
\end{equation*}
Hence $F^{(l)}(\varphi)$ is supported on the thin diagonal of $K$ for all $\varphi\in V$. A standard result from distribution theory, see e.g.\ Proposition $2.3.5$ in \cite{hormander2015analysis}, implies that $F$ has the form \eqref{explicitform} when restricted to $V$, where the $f_{\Vec{\alpha}}(\varphi)$ are \textit{a priori} arbitrary distributions supported on $\pi_1 K$, the projection of $K$ on its first component. To show that they define smooth functions
\begin{equation*}
    f_{\Vec{\alpha}} :  V \to \sD(\pi_1K),
\end{equation*}
we contract the variables $x_2,\ldots, x_l$ with oscillatory functions $e_{\xi_i}$. Combining the fact that the map $\xi \mapsto e_\xi$ is smooth from $\mathbb{R}^n \to \sE$ with the fact that
\begin{equation*}
    F^{(l)} = (\nabla F)^{(l-1)}
\end{equation*}
implies that
\begin{equation*}
    (\varphi,\xi_2,\ldots,\xi_l)    \mapsto   (\varphi,e_{\xi_2},\ldots,e_{\xi_l})   \mapsto  (\nabla F)^{(l-1)}(\varphi)\{e_{\xi_2},\ldots,e_{\xi_l}\}    
\end{equation*}
is a smooth map $V \times (\mathbb{R}^n)^{l-1} \to \sD(\pi_1 K)$. Inserting the explicit expression from equation \eqref{explicitform}, we find that 
\begin{dmath*}
    (\varphi,\xi_2,\ldots,\xi_l) \mapsto \sum_{|\Vec{\alpha}| \leq k} f_{\Vec{\alpha}}(\varphi)(\_) (-i)^{|\Vec{\alpha}|} \prod_{j=2}^l\xi_j^{\alpha_j}e_{\xi_j}(\_)
\end{dmath*}
is a smooth map into $\sD(K)$. The result then follows by taking partial derivatives with respect to the $\xi_i$ and setting all $\xi_i$ to zero.
\end{proof}
\begin{theorem}\label{localfuncts}
Local functionals are equicausal.
\end{theorem}
\begin{proof}
  Let $F$ be a local functional, $\varphi_0\in \sE$, $n\in\mathbb{N}$ and let $V,k,K,$ and $f_\alpha$ be as in the previous proposition. We denote by $Z_n$ the conormal bundle to the diagonal, and note that  $\WF(\partial^{\alpha}\delta_{\Delta_n}) \in Z_n$ for any multi-index $\alpha$. 
  
  Recall that $\mathcal{V}\subset \dot T^*M$ is our choice of cone (generalising the future light cone). From the assumption that the convex hull of $\mathcal{V}$ does not intersect the zero section, it follows that $Z_n \cap \Gamma_n = \emptyset$. Indeed, let $\vec{\xi}=(x,\ldots,x;\xi_1,\ldots,\xi_n) \in Z_n$. By definition of the conormal bundle, it follows that $\sum_{i=1}^n \xi_i = 0$. But then $\vec{\xi}\in \mathcal{V}^n$ would imply that $0$ is in the convex hull of $\mathcal{V}$, which is a contradiction. The case for $(-\mathcal{V})^n$ is similar. 
  
  Hence
  \begin{equation*}
      \varphi \mapsto f_\alpha(\varphi)\mapsto f_\alpha(\varphi)\partial^\alpha\delta_{\Delta_n}
  \end{equation*}
is a smooth map $V\to\sD(K) \to \sD'_{Z_n}(K)$. Summing over $\alpha$, we obtain that $F^{(n)}|_V$ is a smooth map $V \to \sD_{Z_n}(K)$. The result then follows from Proposition \ref{checkec}. 
\end{proof}
\subsubsection{Wick polynomials}
The case for Wick polynomials is comparatively simpler. Recall (see e.g.\ \cite{Baer2009}) that a Wick polynomial is a functional of the form
\begin{equation*}
    P(\varphi) = \sum_{n=0}^k u_n(\varphi^{\otimes n})
\end{equation*}
for $u_n \in \egc{n}$ some fixed distributions, which we take to be symmetric without loss of generality. 
\begin{theorem}
Wick polynomials are equicausal.
\end{theorem}
\begin{proof}
It suffices to show this for monomials as the sum of equicausal functionals is equicausal. Let $P(\varphi)=u(\varphi^{\otimes n})$ for $u\in\egc{n}$. The derivatives of $P$ are, for $l\leq n$
\begin{equation*}
    P^{(l)}(\varphi)\{h_1,\ldots h_l\} = l! u(\varphi^{\otimes n-l}\otimes h_1\otimes \ldots \otimes h_l).
\end{equation*}
Set $\Gamma = \WF (u)$ and choose $K\subset M$ compact such that $\supp (u) \subset K^n$. Then $P^{(l)}$ can be exhibited as the following composition
\begin{equation} \label{composition}
\varphi \mapsto \varphi^{\otimes n-l}\otimes 1^{\otimes l} \mapsto \left( \varphi^{\otimes n-l}\otimes 1^{\otimes l} \right ) \cdot u \mapsto l! \pi_*\left(\left( \varphi^{\otimes n-l}\otimes 1^{\otimes l} \right ) \cdot u\right).
\end{equation}
where $\pi_*$ is the pushforward along the map
\begin{align*}
    \pi:\:&M^n \to M^{l},\\
            &(x_1,\ldots,x_n) \mapsto (x_{n-l+1},\ldots,x_{n}).
\end{align*}
Practically, $\pi_*$ integrates out the first $n-l$ variables. Using Theorem 6.3 of \cite{Brouder2016}, $\pi_*$ is continuous as a map
\begin{equation*}
    \pi_*:\sD'_\Gamma (K^n) \to \sD'_{\pi_*\Gamma} (K^{n-l}) 
\end{equation*}
where
\begin{equation*}
    \pi_* \Gamma = \{(x;\xi) \subset \dot{T}^*M^l \, | \, \exists\, y\in M^{n-l} \textup{ such that } (y,x;0,\xi) \in \Gamma \}.
\end{equation*}
It is a straightforward check that $\Gamma \subset \Gamma_n^c$ implies $\pi_* \Gamma\subset\Gamma_{n-l}^c$. Furthermore, the tensor product of functions is continuous by Theorem 34.1 and Corollary in \cite{Treves1967}. Hence the composition in equation \eqref{composition} can be viewed as a sequence of continuous maps 
\begin{equation*}
    P^{(l)}: \sE(M) \to \sE(M^n) \to \dg{}(K^n) \to \sD'_{\pi_* \Gamma} (K^{n-l}).
\end{equation*}
The result then follows from Proposition \ref{checkec}.
\end{proof}

\section{Time-slice axiom for equicausal functionals}\label{timeslice}
We now show that our definition gives a sub-complex of the multivector fields that is closed under the homotopy operator $\mathcal{H}$. As the primary consequence of this fact, $\sX_\ec$ satisfies the time-slice axiom, by an application of Theorem \ref{quasiinverse}. Recall that the quasi-inverse to the inclusion map is implemented as a pullback by the linear map $\gamma_0$ defined in equation \eqref{gamma0}.

\begin{prop}
If $F\in \sF_\ec$, then $\gamma_0^* F \in \sF_\ec$. Similarly, if $X\in \sX_\ec^l$ and $l\geq 0$, then $\mathcal{H}_l X \in \sX^{l+1}_\ec$. 
\end{prop}
\begin{proof}
To show that $\gamma_0^* F$ is equicausal, we should show that, for any $n \in\mathbb{N}$ and $C\subset \sE$ compact, the set
\begin{equation*}
    \left( \gamma_0^* F \right)^{(n)}(C) \subset \egc{n},
\end{equation*}
and that it is equicontinuous in that space.

We can calculate the $n$'th derivative of $\gamma_0^*F$ explicitly using the chain rule:
\begin{equation*}
    (\gamma_0^* F)^{(n)}(\varphi)\{h_1,\ldots h_n\} = F^{(n)}(\gamma_0\varphi)\{\gamma_0 h_1, \ldots, \gamma_0 h_n\},
\end{equation*}
or, somewhat more succinctly in terms of integral kernels:
\begin{equation}\label{derivativesThetaF}
    (\gamma_0^* F)^{(n)}(\varphi) = (\gamma_0^*)^{\otimes n}  F^{(n)}(\gamma_0\varphi),
\end{equation}
where, on the right, $\gamma_0^*: \sE' \to \sE'$ now also denotes the pullback at the level of distributions by a slight abuse of notation.
 Now denote 
\begin{equation*}
   \gamma_{0,i}^* = \mathord{\underbrace{\mathbbm{1}\otimes\dots\otimes\mathbbm{1}}_{i-1\text{ times}}}\otimes \gamma_0^* \otimes \mathbbm{1}\otimes\dots\otimes\mathbbm{1} ,
\end{equation*}
so that 
 \begin{equation*}
     (\gamma_0^*)^{\otimes n} = \gamma_{0,1}^* \cdots \gamma_{0,n}^*.
 \end{equation*}
 In the proof of Proposition 1 of \cite{Chilian2009}, Chilian and Fredenhagen show that these maps (denoted $\alpha_i$ in their notation) respect the cone $\Gamma_n^c$ . Furthermore, they are continuous as maps
\begin{equation*}
    \gamma_{0,i}^* : \egc{n} \to \egc{n},
\end{equation*}
as contraction with propagators and acting by differential operators are both continuous operations on this space. 

If $C\subset \sE$ is compact, then $\gamma_0 C$ is compact as well, so that $F^{(n)}(\gamma_0 C)\subset \egc{n}$ is equicontinuous by assumption. As composition with continuous maps  respects equicontinuous sets, it follows that 
\begin{equation*}
    (\gamma_0^*)^{\otimes n}  F^{(n)}(\gamma_0 C) \subset \egc{n}
\end{equation*}
is equicontinuous as well, and hence that $\gamma_0^* F$ is an equicausal functional. 

To show that $\mathcal{H}_l X$ is equicausal, we set
\begin{equation*}
    Y(\varphi)\{h_1 \ldots h_{l+1}\} = \int_0^1 X^{(1)}(\gamma_\lambda \varphi)\{\alpha h_1;h_2, \ldots,h_{l+1}\}\lambda^l d\lambda.
\end{equation*}
It suffices to consider only  derivatives of $Y$ as all other terms in $\mathcal{H}_lX$ are related to it by swapping variables. Let $n\in\mathbb{N}$, by Proposition \ref{Leibnizintegralrule}, we can differentiate under the integral $n$ times to find
    \begin{equation}\label{integrateddistr}
        Y^{(n)}(\varphi) = \int_0^1 (\gamma^{*\otimes n}_\lambda \otimes \alpha^*\otimes \mathbbm{1}^{\otimes l})  X^{(n+1)}({\gamma_\lambda\varphi}) \lambda^k d \lambda.
    \end{equation}
 Similar to above, the pullbacks respect the equicontinuous subsets of $\egc{k+n+1}$, so that it is sufficient to show that 
     \begin{equation*}
        \left\{\left.\int_0^1 X^{(n+1)}(\gamma_\lambda \varphi)\lambda^p d\lambda \: \right| \: \varphi \in C \right\} \subset \egc{n+k+1}
    \end{equation*}
is equicontinuous for $C \subset \sE$ compact and $k\leq p \leq k+n $, possible extra factors of $\lambda$ stemming from the $\gamma_\lambda^{*\otimes n}$ factor in equation \eqref{integrateddistr}.  

The set $D = \gamma ([0,1]\times C)$ is compact by continuity of $\gamma$. As $X$ is equicausal, we conclude that $X^{(n+1)}(D)$ is  equicontinuous. Hence there is a continuous seminorm $p$ on $\dg{n+k+1}$ such that
\begin{equation*}
    \left| X^{n+1}(\gamma_\lambda \varphi) (Z)\right| \leq p(Z) \: \forall \lambda\in[0,1],\varphi \in C, Z \in \dg{n+k+1}.
\end{equation*}
Hence we find that
\begin{equation*}
    \left| \int_0^1 X^{n+1}(\gamma_\lambda \varphi) (Z) \lambda^p d\lambda \right| \leq \int_0^1 \left|  X^{n+1}(\gamma_\lambda \varphi) (Z)  \right| d\lambda < p(Z) \: \forall \varphi \in C, Z \in \dg{n+k+1}.
\end{equation*}
which is what we needed to show.
\end{proof}
\begin{theorem}
    The complex $\sX_\ec$ forms a resolution of the space of on-shell equicausal functionals $\iota^*(\sF_\ec)$. The precosheaf defined by mapping a region $N\subset M$ to
\begin{equation*}
    \sX_\ec(N) = \{ X \in \sX_\ec(M) \: | \: \supp(X) \subset N \}.
\end{equation*}
satisfies the time-slice axiom.
\end{theorem}
\begin{proof}
    The previous proposition implies that all the operations in Section \ref{KoszulComplex} respect the sub-complex $\sX_\ec \subset \sX$. Hence all the proofs from that section apply verbatim.
\end{proof}

\section{Star product of equicausal functionals}\label{algebraicstructsection}
Finally, we turn to treating closure of the Poisson bracket and the $\star$-product on the equicausal functionals. They are defined by the following formulae:
\begin{align}
    \left\{F,G\right\}(\varphi) &=\left\langle F^{(1)}(\varphi),\Delta G^{(1)}(\varphi)\right\rangle,  \\ 
    F\star G(\varphi) &= \sum_n \frac{\hbar^n}{n!} \left\langle F^{(n)}(\varphi) , (\Delta^{+})^{\otimes n} G^{(n)}(\varphi)\right\rangle. \label{starprod}
\end{align}
As the equicausal functionals are a subset of the microcausal ones, these maps remain well-defined. We will show that these pairings do give smooth functionals, in contrast to when we act on functionals that are just microcausal. As the Poisson bracket is the lowest order term of the $\star$-commutator we need not check that case separately. The most general case is of course the $\star$-product of multivector fields, as defined in Section 5.2 of \cite{Rejzner2016} for $X\in \sX^p_{\mu c}$, $Y\in\sX^q_{\mu c}$:
\begin{dmath} \label{starproductmultivector}
 \left(X \star Y\right)(\varphi) \{ h_1, \ldots, h_{p+q}\} 
 \quad \doteq \sum_{n=0}^{\infty}\sum_{\sigma \in S_{p+q}} \frac{\hbar^n (-1)^{|\sigma|} }{n ! p ! q !} \left\langle X^{(n)}(\varphi) \{h_{\sigma(1)}, \ldots, h_{\sigma(p)}\},\left(\Delta^{+}\right)^{\otimes n}Y^{(n)}(\varphi) \{ h_{\sigma(p+1)}, \ldots, h_{\sigma(p+q)}\}\right\rangle.
\end{dmath}
This is nothing but the $\star$-product of functionals, combined with a wedge-product of the extra distributional variables. We have elected not to spell out an explicit proof in this scenario, as the extra variables complicate the notation considerably, but add no extra conceptual complication. Futhermore, in the interest of condensing notation, we suppress all vector bundles for the remainder of this section. We start by proving a version of the Leibniz rule.
\begin{prop}[Leibniz rule for equicausal functionals]
    If $W \in \sD'(M^{2n})$ satisfies 
    \begin{equation*}
        \WF(W) \cap (\Gamma_n^c \dot\times \Gamma_n^c) = \emptyset,
    \end{equation*}
and if $F,G$ are equicausal functionals, then the functional
    \begin{equation*}
        W\left(F^{(n)}\otimes G^{(n)}\right): \varphi \mapsto W\left(F^{(n)}(\varphi)\otimes G^{(n)}(\varphi)\right)
    \end{equation*}
    is smooth, and its derivatives are given by
    \begin{equation}\label{Leibniz}
        W\left(F^{(n)}\otimes G^{(n)}\right)^{(m)}(\varphi)\{h^{\otimes m}\} = \sum_{k+l=m} \binom{m}{k} W\left(F^{(n+k)}(\varphi)\{h^{\otimes k},\_\} \otimes G^{(n+l)}(\varphi)\{h^{\otimes l},\_\}\right),
    \end{equation}
    for $h\in \sE$.  
\end{prop}
\begin{proof}
    We have shown that $F^{(n)},G^{(n)}:\sE \to \egc{n}$ are conveniently smooth. The tensor product 
    \begin{equation*}
        \egc{n}\times \egc{n} \to \sE'_{-\WF(W)^c}
    \end{equation*}
 is hypocontinuous, and hence bounded. By Lemma 5.5 in \cite{Kriegl1997}, it is conveniently smooth. Furthermore, $W$ defines a continuous linear map $\sE'_{-\WF(W)^c}\to \mathbb{C}$, which is hence conveniently smooth as well. It follows from the chain rule in the convenient setting, see e.g.\ Theorem 3.18 in \cite{Kriegl1997}, that $W(F^{ (n)}\otimes G^{(n)})$ is conveniently smooth. By Theorem 1 in \cite{Froelicher1982}, it  is also Bastiani smooth, as $\sE(M)$ is a metric space and $\mathbb{C}$ is complete. The Leibniz rule then follows from a standard calculation using the bilinearity of the tensor product.
\end{proof}

This shows that the formal calculations performed in \cite{Brunetti2019} are justified when restricted to equicausal functionals. This does not immediately imply that the resulting functional is again equicausal, merely that it is smooth and that we can use the Leibniz rule to calculate derivatives.
This is the task we will carry out now. For the higher order terms in the $\star$-product, the notation gets somewhat involved due to the required matching of distributional indices. We introduce some custom notation to handle this. Consider the following functional
\begin{dmath}
    \mathcal{B}_n(F,G) (\varphi) = \int F^{(n)}(\varphi) (x_1,\ldots x_n)  \Delta^+(x_1,y_1) \ldots \Delta^+(x_n,y_n) G^{(n)}(\varphi) (y_1,\ldots,y_n) d\vec{x}d\vec{y},
\end{dmath}
which is the $n$'th order term in $F\star G$. If we take an $l$-fold derivative with respect to $\varphi$, we obtain terms of the following form:
\begin{multline}\label{derstarprod}
    \Theta_{n,m,k}(\varphi)\{h_1,\ldots, h_{m+k}\} \\
    \equiv \int F^{(n+m)}(\varphi) (x_1,\ldots ,x_{n+m}) G^{(n+k)}(\varphi)(y_1,\ldots, y_{n+k})\\
     \Delta^+(x_1,y_1) \ldots \Delta^+(x_n,y_n)h_1(x_{n+1})\ldots h_m (x_{n+m})h_{m+1}(y_{n+1})\ldots h_{m+k}(y_{n+k})d\vec{x}d\vec{y}.
\end{multline}
for $m+k=l$. The first $n$ distributional variables of both $F^{(n+m)}$ and $G^{(n+k)}$ get contracted with each other through factors of $\Delta^+$. The remaining variables are contracted against the sections $h_i\in \sE$. 
\begin{figure}[h!]
    \centering
    \includegraphics[width = 0.9\textwidth]{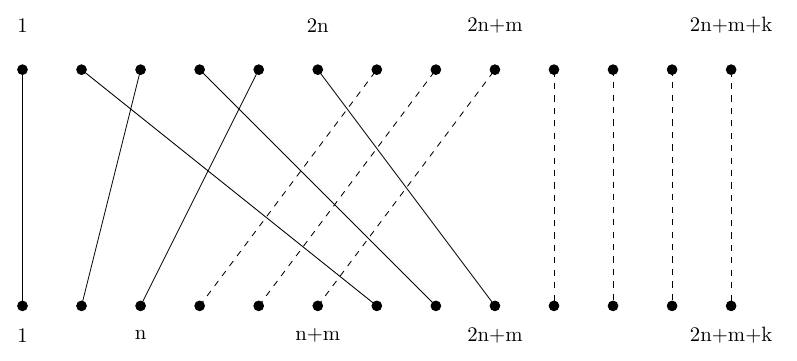}
    \caption{Graphical depiction of $\sigma_{n,m,k}$.}
    \label{graphicalpermutation}
\end{figure}
The following permutation in $S_{2n+m+k}$ keeps track of these arguments:
\begin{equation}\label{disgustingpermutation}
    \sigma_{n,m,k}^{-1}(i) = \begin{cases}
        2i-1            &\textup{ if } 1\leq i \leq n,\\
        i+n             &\textup{ if } n+1\leq i \leq n+m,\\
        2(i-n-m)        &\textup{ if } n+m+1\leq i \leq 2n+m, \\
        i               &\textup{ if } 2n+m+1 \leq i.
    \end{cases}
\end{equation}
We have graphically depicted $\sigma_{n,m,k}$ in Figure \ref{graphicalpermutation}. With this notation, we can simplify equation \eqref{derstarprod} to
\begin{equation}\label{derstarprod2}
\Theta_{n,m,k}(\varphi)\{h_1,\ldots,h_{m+k}\}=\left\langle F^{(n+m)}(\varphi)\otimes G^{(n+k)}(\varphi) , \sigma_{n,m,k}\left( (\Delta^+)^{\otimes n} \otimes h_1\otimes \ldots \otimes h_{m+k}\right)\right\rangle,
\end{equation}
where $\sigma_{n,m,k}$ permutes arguments in the appropriate fashion. We would then like to exchange the factor $h_1\otimes \ldots \otimes h_{m+k}$ for an element of $\dg{m+k}$. The following lemma concerns the wavefront set of such a product.
\begin{lemma}\label{coneinclusion}
The following identity of cones holds:
    \begin{equation}\label{wavefrontproduct}
    \sigma_{n,m,k}(\WF\left((\Delta^+)^{\otimes n}) \dot{\times} \Gamma_{m+k}\right) \cap \Gamma^c_{n+m} \dot{\times} \Gamma^c_{n+k} = \emptyset,
\end{equation}
where $\sigma_{n,m,k}$ is defined in equation \eqref{disgustingpermutation} and the dotted product of cones was introduced in equation \eqref{dotproduct}.
\end{lemma}
\begin{proof}
If $\Xi$ is \textbf{any} cone, we shall denote by $\Tilde{\Xi}$ the union $\Xi \cup \underline{0}$, where $\underline{0}$ is the zero-section in the relevant cotangent bundle. This means that, for cones $\Xi_1, \Xi_2$
\begin{equation*}
 \widetilde{\Xi_1 \dot{\times} \Xi_2} = \tilde{\Xi}_1 \times \tilde{\Xi}_2.   
\end{equation*}
In particular, we have that
\begin{equation*}
    \widetilde{\Gamma_l} = \widetilde{\mathcal{V}}^l \cup \left(-\widetilde{\mathcal{V}}\right)^l
\end{equation*}
for any $l\in\mathbb{N}$. This operation of including the zero section obviously commutes with permuting variables, so that equation \eqref{wavefrontproduct} is equivalent to
    \begin{equation} \label{intersection}
        \sigma_{n,m,k}\left(\widetilde{\WF(\Delta^+)}^{ n}\times\widetilde{\Gamma_{m+k}}\right) \, \cap \, \left(\widetilde{\Gamma^c_{n+m}} \times \widetilde{\Gamma^c_{n+k}}\right) = \underline{0}.  
    \end{equation}
Suppose that $\Vec{\boldsymbol{\xi}}=(\boldsymbol{\xi}_1,\ldots,\boldsymbol{\xi}_{2n+m+k})$ is in this intersection, where $\boldsymbol{\xi}_i  \in T^*M$. The situation is graphically represented in figure \ref{covecfig}, where the dashed covectors belong to $\widetilde{\Gamma^c_{n+m}}$ and the solid covectors belong to $\widetilde{\Gamma^c_{n+k}}$.
\begin{figure}
    \centering
    \includegraphics[width=\textwidth]{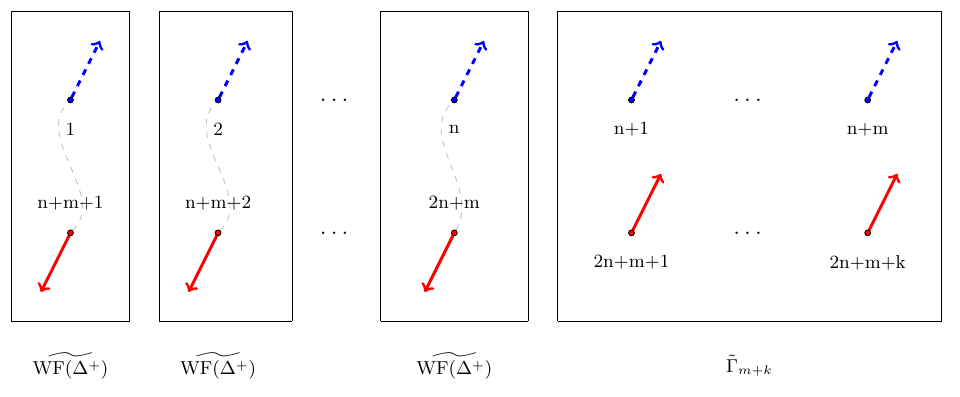}
    \caption{The covectors in Lemma \ref{coneinclusion}.}
    \label{covecfig}
\end{figure}

From the left factor in equation \eqref{intersection}, we see that $\boldsymbol{\xi}_i\in \tilde{\mathcal{V}}\cup -\tilde{\mathcal{V}}$ for all $i$. Furthermore, from the expression of the wavefront set of $\Delta^+$ in equation \eqref{wavefronttwoppoint}, we see that $\boldsymbol{\xi}_1,\ldots, \boldsymbol{\xi}_n \in \tilde{\mathcal{V}}$, and $\boldsymbol{\xi}_{n+m+1},\ldots,\boldsymbol{\xi}_{2n+m} \in -\tilde{\mathcal{V}}$. The remaining covectors are 
    \begin{dmath*}
        (\boldsymbol{\xi}_{n+1}, \ldots, \boldsymbol{\xi}_{n+m},\boldsymbol{\xi}_{2n+m+1},\ldots, \boldsymbol{\xi}_{2n+m+k}) \in \widetilde{\Gamma_{m+k}}.
    \end{dmath*}
Suppose that all of these $\boldsymbol{\xi}_i$ are elements of $\tilde{\mathcal{V}}$ (the case where they are in $-\tilde{\mathcal{V}}$ is analogous). This implies that
    \begin{equation*}
        (\boldsymbol{\xi}_1, \ldots \boldsymbol{\xi}_{n+m}) \in \widetilde{\Gamma_{m+n}}\cap \widetilde{\Gamma^c_{n+m}} = \underline{0}.
    \end{equation*}
In particular, $\boldsymbol{\xi}_1, \ldots, \boldsymbol{\xi}_n \in \underline0$, which implies that $\boldsymbol{\xi}_{n+m+1}, \ldots, \boldsymbol{\xi}_{2n+m} \in \underline0$ as these pairs of covectors are elements of $\widetilde{\WF(\Delta^+)}$. This in turn implies that 
\begin{equation*}
        (\boldsymbol{\xi}_{n+m+1}, \ldots \boldsymbol{\xi}_{2n+l}) \in \widetilde{\Gamma_{k+n}}\cap \widetilde{\Gamma^c_{k+n}} = \underline{0}.
    \end{equation*}
    Hence we conclude that $\Vec{\boldsymbol{\xi}} \in \underline{0}$.
\end{proof}

\begin{prop}\label{starproduct}
If $F$ and $G$ are equicausal functionals and $n\in \mathbb{N}$, then $\mathcal{B}_n(F,G)$ is an  equicausal functional as well.
\end{prop}
\begin{proof}
Let $F,G\in \mathcal{F}_\ec$, $l,n\in\mathbb{N}$ and $C\subset \sE$ compact. By the Leibniz rule \eqref{Leibniz}, the following identity holds
\begin{equation*}
    \mathcal{B}_n(F,G)^{(l)} (\varphi) = \sum_{m+k=l} \binom{l}{m}S_l\Theta_{n,m,k}(\varphi),
\end{equation*}
where $\Theta_{n,m,k}$ is defined in equation \eqref{derstarprod}, and $S_l$ is the operator that symmetrises the distributional variables. Hence it suffices to show that, for all $n,m$ and $k$, $\Theta_{n,m,k}(\varphi)$ admits an extension from $\sE^{m+k}$ to $\dg{m+k}$, which is equicontinuous on $C$. 

We introduce the closed cones
\begin{equation*}
    \Gamma_{n,m,k} \equiv \sigma_{n,m,k}(\WF\left((\Delta^+)^{\otimes n}) \dot{\times} \Gamma_{m+k}\right).
\end{equation*}
Let $\varphi \in \sE$, using Proposition \ref{extensiontensorprod} and Lemma \ref{coneinclusion}, it follows that $F^{(n+m)}(\varphi) \otimes G^{(n+k)}(\varphi)$ admits a continuous extension to $\dg{n,m,k}$, which we denote by $L(\varphi)$. Furthermore, $L(C)$ is an equicontinuous set, as both $F^{(n+m)}(C)$ and $G^{(n+k)}(C)$ are.

Multiplying by factors of $\Delta^+$ and then permuting arguments defines a continuous map
\begin{equation*}
   \Upsilon =\sigma_{n,m,k} \circ \left((\Delta^+)^{\otimes n}\otimes \_ \right) : \dg{m+k} \to \sD'_{\Gamma_{n,m,k}},
\end{equation*}
using hypocontinuity of the tensor product. We denote the restriction to $\sE^{ m+k}$ by the same symbol, so that we can write 
\begin{equation*}
    \Theta_{n,m,k}(\varphi) = L(\varphi) \circ \Upsilon.
\end{equation*}
We obtain the following commutative diagram
\begin{equation*}
\begin{tikzcd}[column sep = large]
\sE^{ m+k} \ar[rd,"\Upsilon "] \ar[dd] \ar[drr, bend left, "\Theta_{n,m,k}(\varphi)"]&\\
&\mathcal{D}'_{\Gamma_{n,m,k}} \ar[r,"L(\varphi)"] & \mathbb{C} \\
\dg{m+k} \ar[ur,"\Upsilon"']  \ar[urr, bend right,"\Upsilon^* L(\varphi)"']
\end{tikzcd}
\end{equation*}
Hence $\Upsilon^* L(\varphi)$ is the sought after extension to $\dg{m+k}$. As $\Upsilon$ is continuous if follows that $\Upsilon^*L(C)$ is equicontinuous, which is what we needed to show.
\end{proof}
We note that, due to the previous proposition at $n=0$, i.e.\ the pointwise product, and Theorem \ref{localfuncts}, we obtain that multilocal functionals are equicausal. More generally, we have shown the following:
\begin{theorem}
    The complex $(\sX_\ec(E)[[\hbar]],\delta)$, equipped with the $\star$-product defined in equation \eqref{starproductmultivector}, is a differential graded algebra. 
\end{theorem}

\section{Conclusions and outlook}\label{outlook}
In this paper, we have critically examined the definition of microcausal functionals. We found that the pointwise fashion in which conditions are imposed results in unwanted behaviour: There exist smooth functionals, whose derivatives have the required singular behaviour at all points in $\sE$, but whose singularities do not vary continuously. As a result, the class of microcausal functionals does not provide a good foundation for constructing models in pAQFT, as they do not close under the $\star$-product and are unlikely to satisfy the time-slice axiom.

We stress that this  is not a consequence of the particular form of the singular structure imposed, or of the fact that we work with open cones, rather than closed ones. Indeed, our counterexample required only regular functionals, which shows that the derivatives of functionals need not have \textbf{any} singularities in order for the Poisson bracket to fail to close. Similarly, counterexamples (not given in this text) to the closure of the homotopy operator $\mathcal{H}$ used in Proposition~\ref{propostionhomoperator} on $\sX_{\mu c}$ can also be phrased in terms of regular multivector fields.  

Our solution to the failure of $\mathcal{H}$ to close was to impose that locally, the singular structure of $F^{(n)}$ should be bounded away from $\Gamma_n$, and that the `coefficients of boundedness' should be uniform when varying $\varphi \in \sE$ slightly. We stress that these extra assumptions are minimal from a practical point of view: The first point is required to avoid covectors approaching the boundary of $\Gamma_n$, which could cause $\WF\left(\int F^{(n)}(\gamma_t) dt \right)$ to fail to be a subset of $\Gamma_n^c$. The second point is needed to make sure that the integrand is bounded, so that the integral can be sensibly performed within $\egc{n}$. It was surprising to us that this minimal solution implies that $F^{(n)}$ is conveniently smooth into $\egc{n}$, which is a far stronger result. It is this fact that also ensures the closure under the Poisson bracket. 

There are other plausible solutions to these problems. One is to use functionals whose derivatives lie in the completion of $\egc{n}$. This modification has been suggested in \cite{Brunetti2019} in order to obtain a space of functionals that is \textit{complete} with respect to a certain topology on the space of field functionals (which is something we have not gone into in the present paper). This space can be characterized as the space of distributions whose \textit{dual} wavefront set, i.e.\ the union of all Sobolev wavefront sets, is contained within $\Lambda$. Functionals whose derivatives lie in these spaces are studied in \cite{Dabrowski2014b,Dabrowski2014c}. We elected not to go into this solution as we preferred working with the smooth wavefront set, which is standard in the pAQFT literature. We expect that similar results to those we obtained in this article can be obtained in that more general setting.

We mention some open problems relating to our present investigation. The first is to extend the present constructions to gauge theories using the BV formalism, as in \cite{Fredenhagen2013,Fredenhagen2012}. Theories with non-trivial local gauge group are never Green hyperbolic, due to the presence of compactly supported solutions to the equations of motion. For this reason, our present results are only applicable after gauge fixing, as that procedure restores Green hyperbolicity. However, it is a priori unclear whether this procedure can always be performed, as it requires a gauge-fixing scheme to be introduced on a model-by-model basis. It would be interesting to see if we could, for free gauge theories, rigorously link the gauge-invariant, on-shell functionals to the homology of the BV complex, without requiring a specific gauge fixing. This problem will be addressed in \cite{HRV2}.

Secondly, in \cite{Brunetti2019}, Brunetti et al.\ studied families of differential operators $P$ arising from the linearisation of the Euler-Lagrange equations around some configuration $\varphi\in\sE$. They imposed that $P$ is normally hyperbolic at each configuration, and that the metric stemming from this operator's principal symbol has smaller light cones than the physical metric of the spacetime $M$, so as not to break causality. 

However, we feel that the class of normally hyperbolic operators is overly restrictive. The Dirac equation, Maxwell's equations, and the Proca equation are not given by normally hyperbolic operators, but are of clear physical interest. For this reason, we have tried to make only minimal assumptions on $P$. The first of these assumptions is Green hyperbolicity, which is the existence of propagators and compatibility with causal structure. This is particularly important for the time-slice axiom. The second assumption is the existence of a cone $\mathcal{V}\subset \dot T^*M$ that governs the wavefront set of the Green's functions. This is used to formulate the generalized Hadamard condition for $\Delta^+$ and to define microcausal and equicausal functionals. However, it is unclear to us how to encode these assumptions into a set of requirements in the Lagrangian formalism.   



Finally, it would be interesting to investigate how we can fit our results in with the program carried out in \cite{Gwilliam2020,Gwilliam2022}, where quantization in terms of nets of algebras is compared to that in terms of factorization algebras. Our equicausal multivector fields are an algebra extending the multilocal functionals used in those sources. In proving their equivalence statement, a central role was played by the time-ordering operator on multilocal functionals. This is an operator that is constructed through Epstein-Glaser renormalisation, and is thus intricately tied to local functionals. It is unclear to us if, and how, this operator might be extended to more general classes of functionals.

\section*{Acknowledgements}
We would like to thank C. Brouder, Y. Dabrowski, C. Fewster, A. Hoffman, A. Riello, and A. Strohmaier for helpful discussions and email exchanges over the course of this project. This research was carried out in part during visits to O. Gwilliam at The University of Massachusetts and to The Perimeter Institute. BV is grateful to  the UKRI for the financial support for his PhD.
\appendixpage
\appendix
\section{Proof of the Leibniz integral rule}
The proof requires the following lemma, which is a more general version of Theorem 2.1.5 in \cite{Hamilton1982}.
\begin{lemma}
    Let $X$ be a topological space and $B$ a complete LCTVS. If $f:X \times [0,1] \to B$ is a continuous map, then 
    \begin{equation*}
        \Tilde{f}(x) = \int_0^1 f(x,s) ds
    \end{equation*}
   defines a continuous map $\Tilde{f}:X \to B$.
\end{lemma}
\begin{proof}
    We note that the integral is well-defined as $B$ is complete and $s\mapsto f(x,s)$ is continuous for fixed $x$. We will show that it is continuous as we vary $x$. 
     
    Let $x_0 \in X$, $H\subset B'$ equicontinuous and $\epsilon > 0$. By definition of equicontinuous sets, there is a neighbourhood of zero $V\subset B$ so that
\begin{equation*}
    |l(b)| < \epsilon \: \forall \: b \in V \: , \: l\in H.
\end{equation*}

Let $W\subset V$ be a \nbhd such that $W+W\subset V$. As $f$ is continuous, there exist, for all $s \in [0,1]$, neighbourhoods $U_s$ of $x_0$ and $I_s$ of $s$ such that 
\begin{equation*}
    x\in U_s, t\in I_s \implies f(x,t)-f(x_0,s) \in W.
\end{equation*}
The sets $U_s\times I_s$ define a cover of $\{x_0\}\times[0,1]$. By compactness, there exists a finite subcover $\{U_i,I_i\}_{i=1}^n$ relating to some $s_1,\ldots, s_n \in [0,1]$. Setting $U=\cap_{i=1}^n U_i$, it follows that, for $x \in U$ and $s\in[0,1]$ 
\begin{equation*}
    f(x,s)-f(x_0,s) = \left( f(x,s)-f(x_0,s_i)\right)+\left(f(x_0,s_i)-f(x_0,s)\right) \in W+W\subset V
\end{equation*}
for some $i$. This means that 
\begin{equation*}
    \left|l\left(\Tilde{f}(x)-\Tilde{f}(x_0)\right)\right| \leq \int_0^1 |l(f(x,s)-f(x_0,s))| ds < \epsilon \: \forall \: l \in H \: , x \in U
\end{equation*}
Hence $\tilde f(x)$ converges to $\tilde f(x_0)$ uniformly on equicontinuous subsets of $B'$ as $x\to x_0$, and therefore also in the original topology on $B$. As the point $x_0$ was arbitrary, it follows that $\tilde{f}$ is continuous.
\end{proof}
\LeibnizIntegralRule*
\begin{proof}
As the partial derivative of $F$ with respect to $\varphi$ exists and is continuous, there is a continuous map $L: \mathbb{R} \times A \times A \times \mathbb{R} \to B$ such that
\begin{equation*}
   L(s,\varphi,\psi,t) = \begin{cases}
       \frac{1}{t}(F(s,\varphi + t \psi) - F(s,\varphi)), \ &t\neq 0 \\
       \frac{\delta}{\delta \varphi} F(s,\varphi)\{\psi\}, \ &t=0.
   \end{cases}
\end{equation*}
We find that, for $t\neq 0$
\begin{align*}
    \frac{1}{t}(G(\varphi + t \psi) - G(\varphi)) &= \int_0^1 \frac{1}{t}(F(s,\varphi + t \psi) - F(s,\varphi))ds \\ &=  \int_0^1  L(s,\varphi,\psi,t)ds \\ &\equiv \tilde{L}(\varphi,\psi,t).
\end{align*}
The map $\tilde{L}$ is continuous by the previous lemma. Hence we can take the limit of $t\to 0$ to conclude that
\begin{equation*}
    G^{(1)}(\varphi)\{\psi\} = \tilde{L}(\varphi,\psi,0) = \int_0^1 L(s,\varphi,\psi,0)ds = \int_0^1 \frac{\delta}{\delta \varphi} F(s,\varphi)\{\psi\}ds,
\end{equation*}
Iterating this argument shows that $G$ is smooth, and that we may perform all derivatives under the integral sign.
\end{proof} \label{integralruleappendix}
\section{Technical details on microlocal analysis}\label{technicallemmata}
This appendix gathers some technical proofs relating to microlocal analysis.

\tensorproductehypo*
\begin{proof}
    Let $v \in \sE'_{\Lambda_b}$ and write $\Gamma = \WF(v)$ and $K=\supp(v)$. We first show that 
    \begin{equation*}
        \_ \otimes v : \: \sE'_{\Lambda_a} \to \sE'_{\Lambda},
    \end{equation*}
is a continuous map.

We use an alternate characterization of the topology on $\sE'_{\Lambda_a}$. Let $\Xi_l \in \Lambda_a$ be an exhausting sequence of \textbf{closed} cones $\Xi_l \in \Lambda_a$ and let $K_l \subset N$ be a compact exhaustion. We refer the reader to Section 3.1 of \cite{Dabrowski2014} for explicit details on this construction. We set
\begin{equation*}
    E_l = \sD'_{\Xi_l}(K_l)\subset \sD'_{\Xi_l}(N),
\end{equation*}
endowed with the subspace topology. Clearly, there are continuous injections between $E_l$ and $E_{l+1}$. It follows that
\begin{equation}\label{directlimitcharact}
    \sE'_{\Lambda_a} = \varinjlim_{l\in\mathbb{N}} E_l
\end{equation}
as a vector-space. By Lemma 10 in \cite{Dabrowski2014}, the inductive limit topology from this diagram matches the strong topology on $\el{a}(N)$, meaning that equation \eqref{directlimitcharact} in fact holds in the category of locally convex topological vector spaces.

The tensor product 
\begin{equation*}
    E_l \xrightarrow{\_\otimes v} \sD_{\Xi_l\dot{\times}\Gamma}(K_l \times K) 
\end{equation*}
is continuous by Proposition \ref{hypocontinuous}, for all $l\in \mathbb{N}$. We obtain the following commutative diagram
\begin{equation*}
\begin{tikzcd}
    E_l \ar[r] \ar[d,"\_\otimes v"] & E_{l+1} \ar[r] \ar[d,"\_\otimes v"] &\el{a}(N) \ar[d,dashed,"\_\otimes v"]\\
    \sD'_{\Xi_l \dot\times \Gamma} \ar[r](K_l \times K) & \sD'_{\Xi_{l+1} \dot\times \Gamma} \ar[r](K_{l+1} \times K) \ar[r] &\sE'_{\Lambda}(N\times P)
\end{tikzcd}
\end{equation*}
in which the horizontal arrows are continuous inclusions. It then follows from the universal property of the direct limit in equation \eqref{directlimitcharact} that the rightmost arrow in the diagram is continuous as well.

Exchanging roles of $N$ and $P$, we have shown that the tensor product
\begin{equation*}
    \el{a}(N) \times \el{b}(P) \to \sE'_{\Lambda}(N\times P)
\end{equation*}
is separately continuous. As both these spaces are barreled, by Proposition 28 in \cite{Dabrowski2014}, it follows from Theorem 41.2 in \cite{Treves1967} that it is in fact hypocontinuous. 
\end{proof}

We turn to proving Lemmas \ref{Curveiscontinuous} and \ref{CurveisSmooth}. The proofs rest on another lemma, which we prove first.
\begin{lemma}\label{fouriertransformsmooth}
If  $u$ is a smooth curve $\mathbb{R}\to \sE'(\mathbb{R}^n)$ for some $n\in\mathbb{N}$, then the Fourier transform of $u$ defines a smooth function
\begin{equation*}
    \hat{u}: \mathbb{R}\times \mathbb{R}^n \to \mathbb{C},
\end{equation*}
given by
\begin{equation*}
    \hat{u}:(t,\xi)\mapsto \hat{u}_t(\xi).
\end{equation*}
\end{lemma}
\begin{proof}
For $\xi\in \mathbb{R}^n$, we define the oscillatory function
\begin{equation*}
    e_\xi(x) = \exp(-i\langle \xi,x\rangle),
\end{equation*}
where $\langle,\rangle$ denotes the Euclidean scalar product on $\mathbb{R}^n$. Suppose that $t_m \to t\in \mathbb{R}$ and $\xi_m \to \xi \in \mathbb{R}^n$, let $\epsilon > 0$ and define the set $B=\{e_{\xi_m} \, | \, m \in \mathbb{N}\}\subset \sE(\mathbb{R}^n)$. As the map $\xi \mapsto e_\xi$ is continuous and $\{\xi_n\, | \, n \in \mathbb{N}\}$ is bounded, $B$ is bounded as well. Due to the fact that $u$ is continuous with respect to the strong topology, we can find a $\delta>0$ such that
\begin{equation*}
    |s-t|<\delta \implies \sup_{f \in B}|u_s(f)-u_t(f)| < \epsilon/2.
\end{equation*}
Also, since $\hat{u}_t$ is a continuous function of $\xi$, there is a $\delta'>0$ such that 

\begin{equation*}
    |\xi - \zeta|<\delta' \implies |u_t(\xi)-u_t(\zeta)| < \epsilon/2.
\end{equation*}
Hence we find, for sufficiently large $n$
\begin{equation*}
    |\hat{u}_{t_n}(\xi_n)-\hat{u}_t(\xi)| \leq |\hat{u}_{t_n}(\xi_n)-\hat{u}_t(\xi_n)| + |\hat{u}_{t}(\xi_n)-\hat{u}_t(\xi)|< \epsilon,
\end{equation*}
This shows that $\hat{u}$ is continuous. To show smoothness, we calculate
\begin{equation*}
    \partial_t \hat{u}_t (\xi) = \partial_t u_t (e_\xi) = \widehat{\partial_t u_t}(\xi).
\end{equation*}
As $\partial_t u_t$ is a smooth curve in its own right, the previous argument implies that this is a jointly continuous function of $t$ and $\xi$. For the partial derivatives with respect to $\xi$, we calculate
\begin{equation*}
    \partial_{\xi_i} \hat{u}_t(\xi) = u_t(- i x_i e_\xi) = - i \widehat{x_i u_t}(\xi).
\end{equation*}
Where we have used the fact that $u_t$ is linear and continuous to move the differentiation inside $u_t$. Again, $x_iu$ is a smooth curve in its own right, and hence $\widehat{x_iu}$ is jointly continuous.

Hence $\hat{u}$ is a $C^1$ function, as both its partial derivatives exist and are continuous. Repeating these steps to higher order shows that $\hat{u}$ is smooth. 
\end{proof}
\Curveiscontinuous*
\begin{proof}
We start with the scalar case, ie. $\tilde{E}=N\times \mathbb{R}$, and show that $t\mapsto u_t(Z)$ is continuous at $0$, as all other points are equivalent. As $\{u_t\}_{t\in(-1,1)}$ is an equicontinuous set, there is a closed cone $\Xi \subset \Lambda$ such that
\begin{equation*}
  \{u_t\}_{t\in(-1,1)} \subset \sD'_\Xi(N), 
\end{equation*}
and is bounded therein. 

We break up the problem using a partition of unity. As $\Xi$ and $\WF(Z)$ are closed sets that are non-intersecting, we can find a collection of opens $U_i$ forming a locally finite cover of $N$, and closed cones $V_i,W_i\subset \mathbb{R}^n$ such that $V_i \cap -W_i = \emptyset$ and 
\begin{align*}
    \Xi &\subset \bigcup_{i} U_i \times V_i^\circ, \\
    \WF(Z) &\subset \bigcup_{i} U_i \times W_i^\circ,
\end{align*}
where the $U_i$ are coordinate patches, so that we can locally trivialize both $\dot{T}^*N$ and $N$ (which we do implicitly to condense notation). We refer the reader to the proof of proposition 14.3 in \cite{Eskin2011} for details on this construction.

Let $\{\psi_i^2\}$ be a partition of unity subordinate to the cover $\{U_i\}$. Then we can write
\begin{equation*}
    u_t(Z) = \sum_i \psi_i^2 u_t(Z) = \sum_i \psi_i u_t(\psi_i Z).
\end{equation*}
As $\bigcup_{t\in(-1,1)}\supp(u_t)$ is compact, only a finite number of terms in this sum contribute. It therefore suffices to show that
\begin{equation*}
    t \mapsto \psi_i u_t(\psi_i Z)
\end{equation*}
is continuous at $0$ for all $i$.  

We write 
\begin{equation}\label{splitintegral}
    \psi_i u_t(\psi_i Z) = \int_{\mathbb{R}_n} \widehat{\psi_i u_t}(\xi)\widehat{\psi_i Z}(-\xi) d\xi = \int_{V_i} \widehat{\psi_i u_t}(\xi)\widehat{\psi_i Z}(-\xi)d\xi + \int_{V_i^c} \widehat{\psi_i u_t}(\xi)\widehat{\psi_i Z}(-\xi)d\xi.
\end{equation}
By the previous lemma, the integrand in both these terms is a smooth function of $t$ and $\xi$. We obtain bounds
\begin{align*}
    |\widehat{\psi_i Z}(\xi)| &\leq   P_{\psi_i,N,V_i}(Z) (1+|\xi|)^{-N} & &\forall \xi \in V_i, & \\
    |\widehat{\psi_i u_t}(\xi)| &\leq \sup_{t\in[-1,1]} P_{\psi_i,N,\overline{V_i^c}}(u_t) (1+|\xi|)^{-N} & &\forall \xi \in V_i^c, t \in (-1,1). &
\end{align*}
Note that the supremum across $t$ is finite as $\{u_t\}_{t\in[-1,1]}$ is bounded in $\sD'_\Xi$. There are also polynomial bounds 
\begin{align*}
    |\widehat{\psi_i Z}(\xi)|     &\leq   C (1+|\xi|)^{M} &&\:\forall \xi \in \mathbb{R}^n, \\
    |\widehat{\psi_i u_t}(\xi)|   &\leq D (1+|\xi|)^{K}   &&\:\forall \xi \in \mathbb{R}^n, t \in (-1,1),
\end{align*}
for some $C,D,M,K$, where we used that $\{u_t\}_{t\in[-1,1]}$ is an equicontinuous subset of $\sE'(N)$ to obtain a uniform bound. It follows that the integrand in both terms of equation \eqref{splitintegral} can be bounded by an integrable function, uniformly in $t$. Hence we can apply the dominated convergence theorem to conclude that
\begin{equation*}
    \lim_{t\to 0} \psi_i u_t(\psi_i Z) = \psi_i u_0(\psi_i Z),
\end{equation*}
as required.

We now treat the case where $\tilde{E}$ is arbitrary. We take a locally finite cover $\{U_i\}$ of $N$ that trivialises $\tilde{E}$ and $\tilde{E}^!$, and select a partition of unity $\{\chi_i^2\}$ subordinate to that cover. This gives isomorphisms
\begin{equation*}
    \sE'_\Lambda(N;\tilde{E}) \xrightarrow{\sim} \bigoplus_i \sE'_{\Lambda|_{U_i}}(U_i;\tilde{E}|_{U_i})\xrightarrow{\sim} \bigoplus_i \sE'_{\Lambda|_{U_i}}(U_i)^k,
\end{equation*}
where $k$ is the rank of $\tilde{E}$. We then apply the first part of this proof to conclude that the curve 
\begin{equation*}
   t\mapsto (\pi_l \chi_i u_t)\left(\pi_l \chi_i Z\right)
\end{equation*}
is continuous, where $\pi_l$ is the projection on the $l$'th factor after acting with the local trivialization of the bundles. As the supports of $u_t$ for $t$ in some bounded interval can be bounded by a compact set $K\subset N$, only finitely many $\chi_i u_t$ are nonzero. It follows that 
\begin{equation*}
    t\mapsto u_t(Z) = \sum_i \sum_{l=1}^k (\pi_l \chi_i u_t)\left(\pi_l \chi_i Z\right),
\end{equation*}
is continuous as well.

For the integration statement, the integral $\int_a^b u_t dt \in \sE'(N;\tilde{E})$ has the defining feature that
    \begin{equation}
        \left(\int_a^b u_tdt\right)(f) = \int_a^b u_t(f) dt \: \: \forall f \in \sE(N;\tilde{E}^!).
    \end{equation}
 We can extend this map to $\sD'_{-\Lambda^c}(N;\tilde{E}^!)$ by setting
\begin{equation*}
        \left(\int_a^b u_tdt\right)(Z) \equiv \int_a^b u_t(Z) dt \: \: \forall Z\in\sD'_{-\Lambda^c},
\end{equation*}
where the right-hand side is well defined by the first part of the lemma. As $\{u_t\}_{t\in [a,b]}$ is an equicontinuous set, there is a continuous seminorm $p$ on $\sD'_{-\Lambda^c}$ so that
\begin{equation*}
    |u_t(Z)| \leq p(Z).
\end{equation*}
Hence we see that 
\begin{equation*}
 \left|\left(\int_a^b u_tdt\right)(Z)\right| \leq \int_a^b|u_t(Z)| dt \leq (b-a) p(Z),  
\end{equation*}
which implies that the extension is continuous. Hence we conclude that
\begin{equation*}
    \int_a^b u_t dt \in \left(\sD'_{-\Lambda^c}\right)' = \el{}.
\end{equation*}
\end{proof}
\CurveisSmooth*
\begin{proof}
 Recall that the topology on $\el{}(N;\tilde{E})$ is that of uniform convergence on bounded sets of $\sD'_{-\Lambda^c}(N;\tilde{E}^!)$. Let $B$ be such a bounded set. By Taylor's theorem, we have the following identity in $\sE'(N;\tilde{E})$:
    \begin{equation*}
        u_t -u_0-t\partial_tu_0 = \int_0^t s \partial^2_su_{s} ds,
    \end{equation*}
where we use the shorthand notation
\begin{equation*}
    \partial_tu_0 := (\partial_t u_t)|_{t=0}.
\end{equation*}
By the previous lemma, applied to the curve $s\mapsto s\partial_s^2u_s$, we see that this is in fact an identity in $\el{}(N;\tilde{E})$: Both sides extend to $\sD'_{-\Lambda^c}(N;\tilde{E}^!)$, and extensions from dense subsets are unique.

By assumption on $u$, there exists continuous seminorm $p$ on $\sD'_{-\Lambda^c}(N;\tilde{E}^!)$ satisfying
\begin{equation*}
    | \partial_s^2u_s(Z)| \leq p(Z) \: \forall s\in (-1,1), Z\in \sD'_{-\Lambda^c}(N;\tilde{E}^!).
\end{equation*}
It follows that
\begin{equation*}
    \left|\frac{1}{t}\left(u_t(Z)-u_0(Z)\right)-\partial_tu_0(Z)\right| \leq \frac{1}{t}\int_0^t |s \partial^2_su_{s}(Z)| ds\leq \frac{t}{2} p(Z).
\end{equation*}
Hence the left hand side converges to zero uniformly on $B$ as $t$ goes to zero, so that $\partial_t u_0$ is the derivative of $u_t$ at $0$ in $\el{}(N;\tilde{E})$. The same argument can be used to show that $u_t$ is differentiable on the whole of $\mathbb{R}$.

As $\partial_tu$ satisfies the same hypotheses as $u$, the result follows from iterating this argument.
\end{proof}
\section*{Declarations}

\subsection*{Data availability statement}
Data sharing is not applicable to this article as no new data were created or analyzed in this study.

\subsection*{Funding and competing interests}
This research was performed as part of BV's PhD, and was supported financially by an EPSRC grant (award number EP/V52010X/1). The authors have no further competing interests to declare that are relevant to the content of this article.

\printbibliography

\end{document}